\documentclass[onecolumn,a4paper]{quantumarticle}
\pdfoutput=1

\usepackage[utf8]{inputenc}
\usepackage{amsmath,amssymb,amsthm}
\usepackage{mathtools}
\usepackage{bm}
\newcommand{\ket}[1]{\lvert #1 \rangle}

\usepackage{graphicx}
\usepackage{booktabs}
\usepackage{pifont}
\usepackage{tikz}
\usetikzlibrary{arrows.meta,calc,decorations.pathmorphing,patterns}
\usepackage[hidelinks]{hyperref}
\usepackage{cleveref}
\crefname{condition}{Condition}{Conditions}
\Crefname{condition}{Condition}{Conditions}
\crefname{assumption}{Assumption}{Assumptions}
\Crefname{assumption}{Assumption}{Assumptions}
\crefname{estimate}{Numerical estimate}{Numerical estimates}
\Crefname{estimate}{Numerical estimate}{Numerical estimates}

\theoremstyle{plain}
\newtheorem{theorem}{Theorem}
\newtheorem{lemma}{Lemma}
\newtheorem{proposition}{Proposition}
\newtheorem{corollary}{Corollary}
\newtheorem{conjecture}{Conjecture}
\theoremstyle{definition}
\newtheorem{definition}{Definition}
\newtheorem{condition}{Condition}
\newtheorem{assumption}{Assumption}
\theoremstyle{remark}
\newtheorem{remark}{Remark}
\newtheorem{estimate}{Numerical estimate}

\newcommand{\Zt}{\mathbb{Z}_2}
\newcommand{\anyon}[1]{\mathsf{#1}}
\newcommand{\ee}{\anyon{e}}
\newcommand{\mm}{\anyon{m}}
\newcommand{\eps}{\anyon{\epsilon}}
\newcommand{\Xbar}{\overline{X}}
\newcommand{\Ybar}{\overline{Y}}
\newcommand{\Zbar}{\overline{Z}}
\newcommand{\Hxy}{H_{XY}}
\newcommand{\Eacc}{E_{A}}
\newcommand{\Ehat}{\widehat{E}_{A}}
\newcommand{\pacc}{p_{\mathrm{acc}}}
\newcommand{\Dstab}{\Delta_{\mathrm{stab}}}
\newcommand{\Braid}{B}

\newcommand{\Qb}{Q_b}
\newcommand{\Pau}{\mathcal P}
\DeclareMathOperator{\dist}{dist}

\DeclareMathOperator{\Tr}{Tr}

\begin{document}

\title{A conditional no-go for resource-free magic-axis measurement on a static surface code}

\author{Jiachen Shen}
\email{jshen28@cougarnet.uh.edu}
\affiliation{Department of Electrical and Computer Engineering, University of Houston, Houston, Texas 77204, USA}
\author{Hui Zhong}
\email{hzhongzg@gmail.com}
\thanks{Corresponding author.}
\affiliation{Department of Electrical and Computer Engineering, University of Houston, Houston, Texas 77204, USA}

\maketitle

\begin{abstract}
Under stated assumptions, a static surface-code patch that adds no fold or \mbox{self-dual} structure cannot perform the
magic-axis check that magic-state cultivation relies on while still accepting often. This is a conditional no-go.
Fault-tolerant machines spend much of their cost making magic states, and cultivation makes them in place by measuring
the magic axis, which every known construction does through a fold or \mbox{self-dual} patch that it is folklore to call
necessary. We test the folklore. The no-go says that a useful check must pay for the magic axis somewhere. It can add a
charge-converting resource, it can leave the dilute regime of its accepted history, or it can accept only exponentially
rarely. For a single stabilizer-measurement transcript this is proved outright, from a topological reading of the accepted
outcome. For adaptive, post-selected protocols in a bounded-depth (polynomial spacetime-volume) model, it holds under two
structural assumptions plus a subcriticality assumption. We isolate the one open assumption, show that protection alone does not force it, and give the threshold
any resolution must address. What remains is a single conjecture.
\end{abstract}

\section{Introduction}\label{sec:intro}

Fault-tolerant quantum computers built from surface codes are efficient at Clifford operations and pay almost all of
their cost for the non-Clifford gates, which are supplied by magic states such as $\ket{T}$~\cite{BravyiKitaev2005,FowlerMariantoni2012}.
Fault tolerance is one layer of a broader effort to make quantum computing practical, alongside distributed quantum
architectures~\cite{Zhong2025UNIQ,Zhong2024DistributedQDP} and the noise, error-correction, and privacy management of
near-term processors~\cite{Zhong2023TuningQEC,Li2023ProjectionDP,Zhao2024QDPInsights,Ju2023HarnessingNoises}; here we
address the fault-tolerant layer, and within it the magic-state bottleneck.
A large-scale computation needs magic states in enormous numbers, and producing each clean one dominates the qubit and
time budget, so any saving in magic-state production is a saving in the whole machine. The standard way to produce
clean magic states is distillation, a long line of protocols that trade many noisy copies for fewer cleaner
ones~\cite{BravyiKitaev2005,BravyiHaah2012,MeierEastinKnill2013,Jones2013,CampbellTerhalVuillot2017}. Distillation is
expensive, and the reason is quantified by the resource theory of magic. That theory measures how far a state sits
from the free stabilizer set and shows that the distance cannot be created by free
operations~\cite{Veitch2014,HowardWallman2017,HowardCampbell2017,SeddonCampbell2019,BravyiGosset2016,HeinrichGross2019}.
The distance is a genuine resource, and a fault-tolerant architecture must pay to import it from outside the free set.
A recent and much cheaper alternative is magic state cultivation, which grows a $\ket{T}$ state in place on a small
surface-code patch and then enlarges the code distance~\cite{GidneyShuttyJones2024,Chen2026}. Cultivation depends on one delicate step. It must
measure the magic axis $\Hxy=(\Xbar+\Ybar)/\sqrt2$, a non-Pauli logical observable whose eigenstates are the
$\ket{T}$-type states. Every surface-code cultivation performs this check through a special piece of structure, a
fold-transversal Hadamard or a self-dual patch, that turns the code back on itself~\cite{Claes2025,Sahay2025,Vaknin2025}.

It is folklore that this structure is necessary, that a plain surface-code patch cannot measure the magic axis on its
own. The reason usually given is topological. In the surface code each logical Pauli is an anyon string. Here $\Xbar$
carries the charge $\mm$, $\Zbar$ carries $\ee$, and $\Ybar$ carries $\eps=\ee\times\mm$. A sharp check of $\Hxy$
superposes $\Xbar$ and $\Ybar$, hence the charges $\mm$ and $\eps$, and these two braid nontrivially with
$\Braid(\mm,\eps)=-1$. A plain patch cannot coherently carry both charges. A fold supplies the missing conversion. This paper tests the
folklore claim directly. It asks whether post-selection can replace the fold on a resource-free patch.

Existing results do not settle this, because they concern the wrong class of operations. Constant-depth
locality-preserving unitaries on a two-dimensional topological code are Clifford~\cite{BravyiKoenig2013}, which already
forbids a unitary magic gate. But a measurement is not a unitary. Cultivation reaches its check through adaptive
measurement, feed-forward, and post-selection, none of which the unitary result covers. Measurement can do more than
shallow unitaries, and it can implement anyon relabelings and boundary flow that unitaries
cannot~\cite{Aasen2023,WebsterBartlett2018}. The closest measurement-regime result is adaptive constant-depth anyon
charge measurement~\cite{BravyiKimKlieschKoenig2022}, which shows that measurement can read out topological charge in
constant depth. Our theorem must be careful to sit beyond it, because it is exactly the capability a skeptic would
point to. Adaptive charge measurement in the abelian surface code reads Pauli, and hence isotropic, charge data. It
does not superpose $\mm$ and $\eps$, so it does not reach the magic axis.
The gap between ``measurement can read charge'' and ``measurement can coherently convert charge'' is precisely where
the magic of $\Hxy$ lives, and it is what a no-go for the magic check has to close.

The physical picture behind the whole argument is a competition between energy and entropy. A coherent $\Hxy$ effect
needs the two braiding charges $\mm$ and $\eps$ to be related across the patch, and on a resource-free patch the only
way to relate them is a connected chain of cells that converts one into the other. Protection makes any such chain
long, of length at least the code distance, so a single chain is costly. But there are exponentially many long chains,
and their entropy can pay for their cost. Whether it does is a phase-transition question, the same question that
decides whether a gas of defects is dilute or condensed. If the accepted, post-selected ensemble keeps the conversion
chains dilute, the coherent effect is exponentially suppressed and the check fails at good acceptance. If the ensemble
condenses the chains, the polymer bound no longer excludes the check, and whether coherent alignment can then be
established is precisely the remaining issue; condensation is an escape route the bound leaves open, not a construction
that succeeds. The theorem reduces the entire question to which side of that transition the accepted ensemble sits on,
and the one input we cannot supply is exactly that the ensemble stays dilute.

We prove a conditional no-go. Under two structural assumptions on the accepted history and one subcriticality
assumption, a static resource-free patch cannot measure the magic axis sharply while accepting often. Subcriticality
is the irreducible core of the problem. The result is a no-go theorem in the standard sense~\cite{Galley2022}. We identify five
conditions, each satisfied by some natural protocol, and prove they cannot hold together. A good cultivation check
wants sharp magic, resource-freeness, protection, and frequent acceptance. The
fifth is subcriticality, a statement that the accepted history does not fill with charge-converting defects. A check
that keeps the first four must give up the fifth: its accepted-effect expansion is forced into a critical or
supercritical patch-spanning conversion sector. This is the effect-level analogue of supplying an anyon-permuting
resource, and the measurement pays for the magic axis in one currency or the other. \Cref{fig:patch} shows the patch and the object at the heart of the argument, and
\cref{app:example} works the whole story on the smallest patch, the distance-3 code.

\begin{figure}[t]
\centering
\begin{tikzpicture}[scale=1.0,>=Latex]
  \draw[thick] (0,0) rectangle (4,4);
  \draw[very thick] (0,0)--(0,4);
  \draw[very thick] (4,0)--(4,4);
  \draw[very thick,dash pattern=on 2pt off 2pt] (0,0)--(4,0);
  \draw[very thick,dash pattern=on 2pt off 2pt] (0,4)--(4,4);
  \node[anchor=south] at (2,4.05) {\footnotesize rough boundary};
  \node[anchor=north] at (2,-0.05) {\footnotesize rough boundary};
  \node[rotate=90,anchor=south] at (-0.15,2) {\footnotesize smooth boundary};
  \node[rotate=-90,anchor=south] at (4.15,2) {\footnotesize smooth boundary};
  \draw[red,very thick] (0,2.6)--(4,2.6);
  \node[red,anchor=south east] at (4,2.62) {\footnotesize $\Xbar\ (\mm)$};
  \draw[blue,very thick] (1.2,0)--(1.2,4);
  \node[blue,anchor=south west] at (1.22,0.1) {\footnotesize $\Zbar\ (\ee)$};
  \draw[teal,very thick,decorate,decoration={snake,amplitude=1.2pt,segment length=6pt}] (0,1.2)--(4,1.2);
  \node[teal,anchor=south,align=center] at (2.6,1.32) {\footnotesize conversion polymer\\$(\mm\!\to\!\eps)$};
\end{tikzpicture}
\caption{The distance-$d$ surface-code patch. The red and blue lines are the native logical strings $\Xbar$ (charge
$\mm$) and $\Zbar$ (charge $\ee$). The teal line marks the missing ingredient. A magic-axis check must turn an $\mm$
string into an $\eps$ string across the patch, and the resource-free patch has no local way to do that.}
\label{fig:patch}
\end{figure}
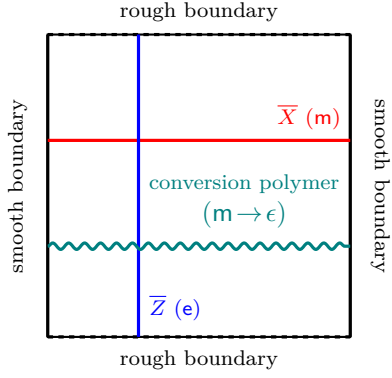

We are explicit about what is proved and what is assumed, because the honesty of a conditional no-go is the whole of
its value. The unconditional core is two statements. A single transcript of resource-free stabilizer measurements has
isotropic anyon charge, so it cannot implement the check; and the magic of an accepted effect equals its acceptance
probability times the distance of its normalised effect from the stabilizer octahedron, a normalisation identity we
prove outright. The first is close to the statement that stabilizer operations cannot measure a non-Pauli observable,
dressed in the language of annular charge that the rest of the argument needs. Neither statement by itself constrains a
decoded family that mixes exponentially many transcripts. The conditional part extends the obstruction to that family.
It combines the normalisation identity with two structural inputs we assume, a local charged-polymer expansion of the
effect and a minimum-length property, together with subcriticality, which makes the sum over exponentially many
patch-spanning polymers converge. Whether subcriticality itself follows from resource-freeness and protection is the
one question we leave open. It does not follow by the natural argument: post-selection reweights a spanning defect
only weakly, while subcriticality needs a constant suppression per unit length. The constructions that would refute it
either use the forbidden resource or leave a coherent-alignment step unestablished. We also compute a spatial-strip
proxy for the entropy scale the assumption must clear, so that it is a concrete target and not a hope, while noting that
the operative spacetime connective constant is not itself computed.

The paper makes three contributions. The first is a framework. We give a topological order parameter for a logical
measurement, the distance of its accepted effect from the stabilizer octahedron, which tells a coherent magic check
apart from a classical imitation that raw charge support cannot. The second is the no-go itself. Under the named
assumptions, a sharp protected check with good acceptance cannot stay both resource-free and subcritical, so it must
add a charge-converting resource, leave the subcritical regime, or accept only exponentially rarely. We state it with
its assumptions named, its unconditional core isolated, and a classification that places every known construction by
the condition it gives up. The third is an open problem. A single line-tension conjecture would remove the one open
assumption, subcriticality, and leave the result resting on the two structural assumptions, and we give a spatial-strip
proxy for the entropy scale any resolution must clear, the operative spacetime threshold itself remaining uncomputed.

The result matters for two reasons beyond the specific check. The first is practical. Cultivation is the cheapest
known way to make magic states, and its whole saving comes from doing the magic measurement in place on a small patch
instead of importing a distilled state. Under the stated assumptions a useful in-place magic
measurement must pay for charge conversion somehow, and known protocols pay by using a fold or a self-dual patch. This
tells a designer that the structure is not an artifact of the current protocols to be optimised away, so effort is
better spent making the fold cheap than making it disappear. The
second reason is conceptual. Most no-go theorems in fault tolerance are about unitaries, and the honest situation for
measurements has been unclear, because measurement is genuinely more powerful and the naive obstructions do not
transfer. Our result shows that a measurement no-go is available, that it takes the form of a resource-necessity
statement with an order parameter and a phase-transition assumption, and that the assumption is the real physics of the
problem. This is a template one can carry to other non-Clifford operations and other codes, and the method, reading a
logical operation against the free polytope and asking which defect sectors an accepted history can support, is not
special to $\Hxy$.

The paper is organised as follows. \Cref{sec:prelim} fixes the surface-code block, the description of an adaptive
protocol as an accepted effect, and the two definitions the theorem rests on, the magic witness and the resource-free
class. \Cref{sec:results} states the five conditions, proves the fine-grained isotropy theorem and the conditional
suppression theorem, assembles them into the no-go, and records what each condition contributes. \Cref{sec:classification}
reads the theorem as a classification of prior constructions and places it relative to the earlier no-go theorems for
topological codes. \Cref{sec:open} examines the one open assumption, states the line-tension conjecture, and explains
why the natural arguments for it are circular. \Cref{sec:discussion} separates what is proved from what is assumed and
from what is not claimed. The appendices carry the full proofs. \Cref{app:thm1} proves the isotropy theorem and its annular-charge lemmas.
\Cref{app:thm2} proves the conditional suppression theorem with its normalisation and polymer bounds.
\Cref{app:linetension} analyses the line-tension conjecture and the closest would-be counterexample. \Cref{app:numerics}
gives the spatial-strip proxy for the entropy scale. A worked distance-3 example runs through \cref{app:example}. \Cref{app:localclass}
carries out the local-noise reduction, deriving the charged-polymer and cleaning content of the structural assumptions
from locality and isolating the residual subcriticality hypothesis as a checkable property of the decoder.

\section{Setup and definitions}\label{sec:prelim}

This section fixes the object we study and states the two definitions the no-go rests on. We work with one fixed
piece of hardware, a single planar surface-code patch that stores one logical qubit and is never deformed. On this
patch we ask whether a magic logical measurement can be carried out using only operations that add no topological
structure. To make that question precise we need three ingredients. First, we fix the anyon content of the patch and
the logical operator we want to measure (\cref{sec:prelim-block}). Second, we describe an arbitrary adaptive,
post-selected protocol as a single accepted effect (\cref{sec:prelim-effects}). Third, we define a witness that tells
a genuine magic measurement apart from a classical imitation (\cref{sec:prelim-witness}). We close with the definition
of the resource-free operations the theorem is about (\cref{sec:prelim-resourcefree}).

We use standard surface-code facts without reproving them. The stabilizer formalism describes the code and its logical
operators~\cite{Gottesman1997,AaronsonGottesman2004,DehaeneDeMoor2003}. The toric-code phase and its gapped boundaries
describe the anyons, the boundaries, and the edge condensates~\cite{Kitaev2003,KitaevKong2012,BeigiShorWhalen2011,Barkeshli2019}. Protection and the acceptance decision
rest on the topological-memory and decoding literature~\cite{Dennis2002,DuclosCianciPoulin2010,DelfosseNickerson2021}.
The operations the theorem excludes are exactly the ones that other fault-tolerant constructions use to move logical
information, namely lattice surgery~\cite{Horsman2012,RaussendorfHarringtonGoyal2007,LitinskiVonOppen2018,Litinski2019}
and the anyon-permuting dynamics of Floquet and dynamical codes~\cite{HastingsHaah2021,Vuillot2021,Davydova2023}; the
resource-free class is defined against these, and \cref{sec:classification} places each of them relative to the
theorem.

\begin{table}[t]
\centering
\caption{Notation used throughout.}
\label{tab:notation}
\resizebox{\textwidth}{!}{%
\begin{tabular}{ll}
\toprule
Symbol & Meaning \\
\midrule
$d$ & code distance of the planar surface-code patch \\
$\{1,\ee,\mm,\eps\}$ & toric-code anyons, $\eps=\ee\times\mm$; braiding $\Braid$, $\Braid(\mm,\eps)=-1$ \\
$\Xbar,\Ybar,\Zbar$ & logical Paulis; Wilson strings of charge $\mm,\eps,\ee$ \\
$\Hxy=(\Xbar+\Ybar)/\sqrt2$ & target magic axis; L\"uders projectors $(I\pm\Hxy)/2$ \\
$K_\tau,\ E_\tau=K_\tau^\dagger K_\tau$ & Kraus operator and effect of a fine transcript $\tau$ \\
$A,\ \Eacc=\sum_{\tau\in A}K_\tau^\dagger K_\tau$ & accepted family and its coarse-grained effect \\
$\pacc,\ \Ehat$ & acceptance probability and normalised logical effect \\
$\Dstab(\Ehat)$ & magic witness: distance to the stabilizer-effect octahedron \\
$\Qb$ & code-preserving annular charge content of the effect \\
$L(d)$ & instrument-level fault distance \\
$G_d,\ |V(G_d)|$ & accepted spacetime cell graph and its (polynomial) volume \\
$\mu_G,\ z,\ z_c=1/\mu_G,\ \tau$ & spacetime conversion-polymer connective constant, activity, critical activity, line tension \\
\bottomrule
\end{tabular}}
\end{table}

\subsection{The surface-code block and the target measurement}\label{sec:prelim-block}

We fix a planar rotated surface code encoding one logical qubit, with fixed rough and smooth boundaries and code
distance $d$~\cite{Kitaev2003,BravyiKitaev1998,FowlerMariantoni2012,BombinMartinDelgado2007}. Data qubits sit on the
faces of a rotated square lattice, and the stabilizer group is generated by weight-four $X$-type and $Z$-type checks in
the bulk, with weight-two checks along the boundaries. The two rough boundaries absorb $\ee$ charge and the two smooth
boundaries absorb $\mm$ charge, and this choice of condensates is what pins the logical operators to definite anyon
type. The distance $d$ is the least weight of a logical operator, and it is the length of the shortest string
connecting a pair of like boundaries.

The bulk realises the toric-code phase $D(\Zt)$, whose anyons are $\{1,\ee,\mm,\eps\}$ with $\eps=\ee\times\mm$. The
braiding form $\Braid$ is nondegenerate, and the only nontrivial values we use are
$\Braid(\ee,\mm)=\Braid(\mm,\eps)=\Braid(\ee,\eps)=-1$; every charge braids trivially with itself and with the vacuum.
The logical Pauli operators are Wilson lines of definite anyon type. $\Xbar$ is an $\mm$-string between the two smooth
boundaries, $\Zbar$ is an $\ee$-string between the two rough boundaries, and $\Ybar=i\Xbar\Zbar$ is an $\eps$-string
along the diagonal. The single logical anticommutation $\Xbar\Zbar=-\Zbar\Xbar$ is the lattice image of the toric-code
braiding $\Braid(\ee,\mm)=-1$, and it is the one algebraic fact the whole obstruction turns on.

The logical measurement we target is the nondemolition L\"uders check of the magic axis
\begin{equation}\label{eq:hxy}
  \Hxy=\frac{\Xbar+\Ybar}{\sqrt2},\qquad
  \text{with L\"uders projectors}\quad \frac{I\pm\Hxy}{2}.
\end{equation}
This is the check that a surface-code cultivation of the $\ket{T}$ state needs~\cite{GidneyShuttyJones2024,Claes2025,Sahay2025}.
The hard part is coherence. An $\Hxy$ eigenstate keeps a fixed phase between the $\mm$-string sector and the
$\eps$-string sector. A classical choice between them is not enough. Because $\Braid(\mm,\eps)=-1$, these two sectors braid nontrivially, so a coherent
$\Hxy$ effect requires the patch to carry $\mm$ and $\eps$ charge at once. A patch that can only carry a set of
mutually transparent charges cannot host it. \Cref{app:example} works the whole story on the smallest patch, the
distance-3 code, and is a concrete instance to keep in mind while reading the definitions.

\subsection{Adaptive protocols as accepted effects}\label{sec:prelim-effects}

A cultivation protocol measures, adapts, and then accepts or rejects. It interleaves measurements, classical
feed-forward, and a final acceptance decision, and only accepted runs produce output. We describe it at two grains. A
single \emph{fine transcript} $\tau$ is one complete record of measurement outcomes and feed-forward choices. It acts
on the code space by a Kraus operator $K_\tau$, and its accepted effect is $E_\tau=K_\tau^{\dagger}K_\tau$, the positive
operator for which $\langle\psi|E_\tau|\psi\rangle$ is the probability that the run follows transcript $\tau$ and is
accepted, starting from $|\psi\rangle$. A decoder groups fine transcripts into an \emph{accepted family} $A$, the set
of transcripts the protocol declares successful. The object the reader should keep in mind throughout is the
coarse-grained effect of that family,
\begin{equation}\label{eq:accepted-effect}
  \Eacc=\sum_{\tau\in A}K_\tau^{\dagger}K_\tau ,
\end{equation}
together with its acceptance probability $\pacc=\tfrac12\Tr(P_{\text{code}}\Eacc P_{\text{code}})$ and its
normalised logical effect $\Ehat=P_{\text{code}}\Eacc P_{\text{code}}/(2\pacc)$. We must work with $\Eacc$, not
any single $K_\tau$. A decoder can collapse exponentially many transcripts into one high-acceptance event, and a
statement proved transcript by transcript can fail for the family they form.

To read the anyon content of an accepted branch we use the annular charge of its effect. Let $b$ be an annulus
that wraps the logical string. Each Pauli operator $P$ supported near $b$ has an annular charge
$d_b(P)\in\{1,\ee,\mm,\eps\}$, read from the homology classes of its $X$- and $Z$-strings after cleaning by code
stabilizers and contractible neutral dressings (\cref{app:thm1}). The \emph{annular charge content} of the effect is
the set $\Qb(A)=\{d_b(s)\}$ of charges carried by its code-preserving Pauli support.
We call the annular charge content \emph{isotropic} when all its charges braid trivially with one another, so that
$\Qb(A)$ lies in a transparent subgroup like $\{1,\mm\}$ or $\{1,\ee\}$. The coherent $\Hxy$ effect of
\cref{eq:hxy} needs the non-isotropic pair $\mm,\eps$. A subtlety the proof must respect is that the charge is read on
\emph{cleaned, code-preserving} representatives, not on the raw Pauli support. \Cref{app:thm1} shows that raw support
can carry a non-isotropic pair while acting trivially on the code space, so $\Qb(A)$ is defined through the code
projector.

One more point fixes the level at which we work. A protocol defines a full quantum instrument, the map that sends the
input to the accepted output together with the classical acceptance flag. The instrument carries more information than
its effect $\Eacc$. It records the post-measurement state, whereas $\Eacc$ records only the accepted outcome
statistics, so the effect does not by itself determine the output state. We prove an obstruction at the level of the
effect. This is a \emph{necessary condition} for the instrument: a nondemolition $\Hxy$ L\"uders measurement has a
definite accepted effect, the projector $(I+\Hxy)/2$, so ruling out that effect rules out any instrument that would
realise the measurement. The obstruction is therefore weaker than a full instrument-level characterisation, and we do
not claim the effect captures the output state that cultivation ultimately wants. It captures exactly the part a
coherent check must get right, and the witness below acts at the same level.

\subsection{The magic witness}\label{sec:prelim-witness}

Not every effect with $\Xbar$ and $\Ybar$ support is a magic measurement, so we need a witness that is not fooled
by classical imitations. The right object is the distance of the normalised logical effect from the set of effects
that native Pauli measurements can build. For one logical qubit, write $\Ehat=\tfrac12(I+\bm v\cdot\bm\sigma)$ with
Bloch vector $\bm v=(v_x,v_y,v_z)$. The native operations are the six Pauli-Wilson-measurement projectors
$\tfrac12(I\pm\Xbar)$, $\tfrac12(I\pm\Ybar)$, $\tfrac12(I\pm\Zbar)$, whose Bloch vectors are the six unit vectors
$\pm\hat e_x,\pm\hat e_y,\pm\hat e_z$. A protocol that measures Pauli axes and mixes the outcomes classically produces
a convex combination of these, so its normalised effect has $\bm v$ in their convex hull, the stabilizer-effect
octahedron $\{|v_x|+|v_y|+|v_z|\le 1\}$. This octahedron is the one-qubit face of the resource theory of
magic~\cite{BravyiKitaev2005,Veitch2014,HowardCampbell2017}, the same polytope of free effects whose extreme points are
the stabilizer measurements. An effect outside it cannot be assembled from Pauli measurements and classical mixing, so
a positive distance to it is a certificate that the check does something no classical combination of Pauli readouts
can.

\begin{definition}[Magic witness]\label{def:dstab}
The \emph{magic} of a normalised one-qubit logical effect $\Ehat=\tfrac12(I+\bm v\cdot\bm\sigma)$ is its
$\ell_1$-distance to the stabilizer-effect octahedron,
\begin{equation}
  \Dstab(\Ehat)=\dist\!\bigl(\bm v,\ \{\,\bm w:\|\bm w\|_1\le1\,\}\bigr).
\end{equation}
\end{definition}

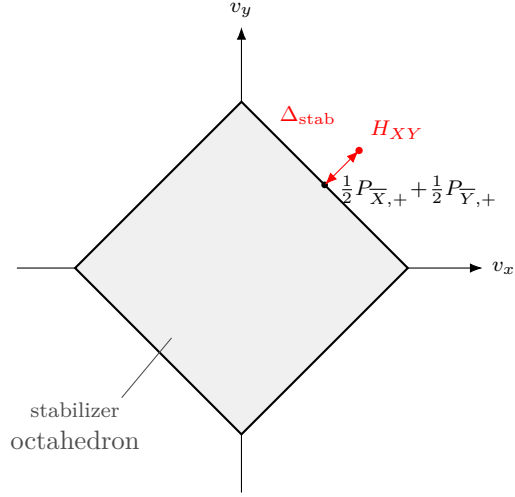
\begin{figure}[t]
\centering
\begin{tikzpicture}[scale=2.2,>=Latex]
  \draw[->] (-1.35,0)--(1.45,0) node[right] {\footnotesize $v_x$};
  \draw[->] (0,-1.35)--(0,1.45) node[above] {\footnotesize $v_y$};
  \draw[thick,fill=gray!12] (1,0)--(0,1)--(-1,0)--(0,-1)--cycle;
  \node[gray!60!black,align=center] at (-1.0,-0.95) {\footnotesize stabilizer\\octahedron};
  \draw[gray!60!black,thin] (-0.72,-0.78)--(-0.42,-0.42);
  \fill (0.5,0.5) circle (0.02);
  \node[anchor=west] at (0.53,0.46) {\footnotesize $\tfrac12P_{\Xbar,+}\!+\!\tfrac12P_{\Ybar,+}$};
  \fill[red] (0.707,0.707) circle (0.022);
  \node[red,anchor=south west] at (0.72,0.72) {\footnotesize $\Hxy$};
  \draw[red,<->] (0.5,0.5)--(0.707,0.707);
  \node[red] at (0.40,0.92) {\footnotesize $\Dstab$};
\end{tikzpicture}
\caption{The $v_z=0$ slice of the stabilizer-effect octahedron. The sharp $\Hxy$ projector (red) lies outside it at
distance $\Dstab=\sqrt2-1$; the classical mixture $\tfrac12P_{\Xbar,+}+\tfrac12P_{\Ybar,+}$ sits on the boundary
edge, with $\Dstab=0$.}
\label{fig:octahedron}
\end{figure}

The witness does exactly the job we need. The octahedron $\{\|\bm w\|_1\le1\}$ has eight facets, one for each choice of
signs $(\pm,\pm,\pm)$, and in the octant of the sharp check the nearest facet is $w_x+w_y+w_z=1$. The $\ell_1$-distance
of a point $\bm v$ outside it is the excess of the $\ell_1$ norm over one,
\begin{equation}
  \Dstab\bigl(\tfrac12(I+\bm v\cdot\bm\sigma)\bigr)=\max\bigl(0,\|\bm v\|_1-1\bigr).
\end{equation}
The sharp $\Hxy$ projector has $\bm v=(1/\sqrt2,1/\sqrt2,0)$, hence
\begin{equation}
  \|\bm v\|_1=\tfrac{1}{\sqrt2}+\tfrac{1}{\sqrt2}=\sqrt2>1,\qquad \Dstab=\sqrt2-1>0,
\end{equation}
so it lies strictly outside the octahedron (\cref{fig:octahedron}). A coin flip between an $\Xbar$ measurement and a
$\Ybar$ measurement gives $\tfrac12P_{\Xbar,+}+\tfrac12P_{\Ybar,+}$ with $\bm v=(\tfrac12,\tfrac12,0)$ and
$\|\bm v\|_1=1$, so it sits on the octahedron boundary and $\Dstab=0$. The two effects have the same Pauli support,
along $\Xbar$ and $\Ybar$, and the witness still tells them apart, because it sees the $\ell_1$ length of the Bloch
vector and not merely which axes are present. This is the property the first condition of the no-go will name.

The witness is built from the same free set as the resource theory of magic. The stabilizer-effect octahedron is the
convex hull of the native Pauli-measurement effects, and $\Dstab$ is the distance to it, so $\Dstab$ vanishes exactly
on that free set and is positive only for an effect no classical mixture of Pauli measurements can build. It is the
effect-side analogue of the state quantities that measure how far a state is from the stabilizer
polytope~\cite{Veitch2014,HowardCampbell2017,SeddonCampbell2019,HeinrichGross2019}. We use it as a witness and order
parameter, not as a proven monotone: we do not claim $\Dstab$ is contractive under the full instrument class of
\cref{def:resourcefree}, only that it separates a coherent check from a classical imitation. That separation is what
makes it non-circular. It is defined by the native operations, not by the conclusion we want to reach, and a protocol
whose accepted effect lands outside the octahedron has demonstrably done more than measure and mix Pauli axes.

\subsection{Resource-free operations}\label{sec:prelim-resourcefree}

The theorem concerns protocols that add no anyon-permuting structure to the patch. We state this as a property of
the whole spacetime history, not of any single gate, because a protocol can smuggle in structure through ancillas,
selectors, or the decoder. The definition is geometric: it constrains the topology and the ancilla content, and it
allows arbitrary quantum instruments as long as no single cell converts charge on its own.

\begin{definition}[Resource-free protocol]\label{def:resourcefree}
A protocol on the fixed distance-$d$ patch is \emph{resource-free} when its accepted spacetime history is a product
cobordism of the same $D(\Zt)$ phase with fixed rough/smooth condensates, and it satisfies all of:
\begin{enumerate}
  \item[(i)] \emph{No anyon-permuting structure.} No $\ee\!\leftrightarrow\!\mm$ domain wall, twist, fold,
        cross-cap, graft, code switch, self-dual patch, or Floquet excursion; the bulk measurement-quantum-cellular-automaton
        index is trivial and the boundary carries no anomalous GNVW flow~\cite{BravyiKoenig2013,Aasen2023,WebsterBartlett2018}.
  \item[(ii)] \emph{Free ancillas.} Ancillas are product stabilizer states; no magic states, catalysts, or encoded
        ancilla patches are supplied.
  \item[(iii)] \emph{Charge-non-converting cells.} Microscopic instruments may be arbitrary completely positive maps
        with post-selection, but each single cell preserves the local charge sectors. Formally, writing
        $\{\Pi^{(x)}_c\}_{c\in\{1,\ee,\mm,\eps\}}$ for the local charge-sector projectors on the neighbourhood of a cell
        $x$, every Kraus operator of the cell commutes with each $\Pi^{(x)}_c$, so it has no matrix element between
        distinct sectors and no cell implements a coherent $\mm\!\leftrightarrow\!\eps$ conversion (a magic selector).
        This allows non-Clifford charge-preserving filters and native Pauli-Wilson measurements; only the coherent
        cross-sector element is excluded.
\end{enumerate}
\end{definition}

Protection is a separate property, stated as its own condition later, so that ``resource-free'' names only the three
clauses above. We call a resource-free protocol \emph{protected to $L(d)$} when its whole accepted operation, including
selectors, acceptance logic, and decoding, has spacetime fault distance $L(d)$.

Two comments fix the reading of \cref{def:resourcefree}. First, clause (iii) is the setting in which the no-go is a
statement, not a triviality. If we forbade all nonstabilizer instruments, native Pauli operations plus stabilizer
ancillas could not build a nonstabilizer effect at all, and the whole question would collapse. If we allowed a single
cell to convert $\mm$ into $\eps$, that cell would already be the magic measurement we are asking about.
Charge-non-converting broad instruments are the honest middle. Each cell is free of the resource, and the question is
whether many such cells plus post-selection can nonetheless assemble a coherent conversion across the patch. The
class is genuinely broad: it contains every constant-depth Clifford circuit, every local Pauli measurement pattern,
arbitrary classical feed-forward and randomness, and any post-selection on the measured record, all composed to any
depth. What it excludes is a single physical primitive, the local coherent charge converter, and the theorem is the
claim that no amount of the allowed operations reconstructs that primitive globally without paying in acceptance or
distance. \Cref{def:resourcefree} itself allows arbitrary finite depth; the no-go of \cref{thm:tradeoff,thm:main},
however, applies to the bounded-degree, polynomial-volume subfamily $|V(G_d)|\le\mathrm{poly}(d)$ defined in
\cref{app:thm2}, so a superpolynomial-depth protocol is outside its reach. We flag this here so the depth restriction of
the theorem is not mistaken for a restriction on the class.

We make the local charge-sector projectors of clause~(iii) concrete at the lattice level. Around each cell $x$ fix a
small disk $D_x$ and let $\ell_x=\partial D_x$ be its boundary loop. The closed Wilson and parity loop operators on
$\ell_x$, equivalently the product of the stabilizer syndromes enclosed by $\ell_x$, commute, and their joint
eigenvalue is a projective measurement of the total anyon charge enclosed, with outcomes labelled by
$c\in\{1,\ee,\mm,\eps\}$; $\Pi^{(x)}_c$ is the projector onto outcome $c$. These are the standard superselection-sector
projectors of the $D(\Zt)$ phase, and ``charge-preserving'' means each Kraus operator of the cell commutes with the
syndrome of every such enclosing loop, so it maps each sector to itself. For a cell whose neighbourhood meets a gapped
boundary, where the condensate makes the full $\{1,\ee,\mm,\eps\}$ labelling degenerate, $\Pi^{(x)}_c$ is understood
after thickening $D_x$ to a bulk disk before condensation, equivalently as the relative half-annular charge-sector
projector compatible with the boundary condensate; boundary-changing maps are separately excluded by clause~(i). This is why our notion is
\emph{topological-charge-conversion-freeness}, not the stabilizer resource theory's notion of a free operation on the
logical qubit: a cell may be an arbitrary, even non-Clifford, completely positive map, provided it does not connect
distinct anyon sectors. Such a filter is block-diagonal in the sector labels, so on its own it cannot create the
off-diagonal $\mm\!\leftrightarrow\!\eps$ matrix element a magic effect needs, and by \cref{thm:iso} each accepted
stabilizer transcript it takes part in has isotropic annular charge. Whether many such cells with post-selection can
nonetheless assemble that coherence globally is exactly the question the theorem answers, and it is what
\cref{ass:decomp,ass:cleaning} and \cref{cond:subcrit} control.

Second, the clauses are not redundant. Clause (i) is a statement about the spacetime topology and the boundary flow,
clause (ii) about the ancilla resource, and clause (iii) about the microscopic instruments. A protocol can respect the
topology and still smuggle in the resource through a magic ancilla, which clause (ii) blocks, or through a coherent
local filter, which clause (iii) blocks. The three together are what make ``resource-free'' a property of the whole
history and not just of its coarse geometry.

\section{The no-go: five incompatible conditions}\label{sec:results}

We now state the result. It is a no-go theorem in the standard sense~\cite{Galley2022}. The theorem names five natural
conditions and proves that they cannot all hold at once. Four of the five are properties one wants a
good magic measurement to have. The fifth is a statistical-mechanical property of the accepted history. The theorem
says that a resource-free, protected, high-acceptance protocol which measures the magic axis sharply must give up
the fifth. We first name the conditions (\cref{sec:results-conditions}), then state the incompatibility and the two
theorems that establish it (\cref{sec:results-theorems}).

\emph{Main conditional theorem (informal statement of \cref{thm:tradeoff,thm:main}).} Fix a bounded-degree accepted
spacetime cell graph $G_d$ of polynomial volume $|V(G_d)|\le\mathrm{poly}(d)$. Suppose the accepted effect admits the
cone-compatible charged expansion of \cref{ass:decomp}, satisfies the effect-cleaning property of \cref{ass:cleaning},
and its accepted conversion sector is subcritical, $\sum_{|\Gamma_\gamma|=\ell}\|B_\gamma\|_c\le C_0|V(G_d)|(\mu_G
z)^\ell$ with $\mu_G z<1$. Then
\begin{equation*}
  \Dstab(\Ehat)\,\pacc\le C\,|V(G_d)|\,e^{-\tau L(d)},\qquad \tau=-\log(\mu_G z)>0 .
\end{equation*}
Hence if $L(d)=\Omega(d)$ and $|V(G_d)|=\mathrm{poly}(d)$, sharp magic $\Dstab(\Ehat)\ge c>0$ forces
$\pacc\le e^{-\Omega(d)}$: a resource-free, protected, high-acceptance sharp magic-axis check must give up subcriticality.
The associated \emph{line-tension conjecture} (\cref{conj:linetension}) is that, for resource-free, protected,
non-negligibly accepting families in this bounded spacetime model, the accepted conversion sector is indeed subcritical,
$\mu_G z<1$; we do not prove it, and \cref{sec:open,app:linetension} isolate it as the one open input, with
\cref{app:localclass} giving its concrete form for local noise. The unconditional core, needing none of the assumptions,
is the single-transcript isotropy theorem \cref{thm:iso} with the normalisation identity \cref{lem:norm}.

\subsection{The five conditions}\label{sec:results-conditions}

Each condition is a property of a coarse-grained accepted family $A$ on the fixed distance-$d$ patch. We state each
one plainly before the formal line.

This condition says the accepted effect is genuinely magic, not a random Pauli check.
\begin{condition}[Sharp magic]\label{cond:magic}
The normalised accepted effect stays a fixed distance outside the stabilizer-effect octahedron: there is a constant
$c>0$, independent of $d$, with $\Dstab(\Ehat)\ge c$. The sharp $\Hxy$ check is the motivating case, with
$c=\sqrt2-1$.
\end{condition}

This condition says the protocol adds no anyon-permuting structure.
\begin{condition}[Resource-free]\label{cond:free}
The protocol is resource-free in the sense of \cref{def:resourcefree}.
\end{condition}

This condition says the whole protocol is protected to linear order.
\begin{condition}[Protection]\label{cond:protect}
The accepted operation, including selectors, acceptance logic, and decoding, has spacetime fault distance
$L(d)=\Omega(d)$.
\end{condition}

This condition says the protocol succeeds often enough to be useful.
\begin{condition}[Acceptance]\label{cond:accept}
The family is accepted with non-negligible probability: $\pacc\ge 1/\mathrm{poly}(d)$, counted honestly over the
whole family, with retries and discards included in the count.
\end{condition}

This condition says accepted histories do not fill with a dense gas of charge-conversion paths.
\begin{condition}[Subcriticality]\label{cond:subcrit}
The accepted charge-conversion sector has positive line tension. Writing the accepted weight of $\mm\!\leftrightarrow\!\eps$
conversion polymers in the accepted spacetime cell graph $G_d$ as a gas of self-avoiding conversion spines with
renormalized per-cell activity $z$ (into which the branch and decoration multiplicity off each spine is summed) and
connective constant $\mu_G$, the gas is subcritical. That is $\mu_G z<1$, equivalently there is a constant $\tau>0$ with
per-length weight bounded by $e^{-\tau\ell}$.
\end{condition}

\Cref{cond:magic,cond:free,cond:protect,cond:accept} are the properties a useful cultivation check is built to have.
\Cref{cond:subcrit} is different in kind. It is a free-energy statement about the ensemble of accepted histories, and
whether it follows from the others is the one question we cannot settle. We treat it as a named condition here and
return to its status in \cref{sec:open}.

\begin{remark}[Scope of the magic condition]\label{rem:magic-scope}
\Cref{cond:magic} asks only for a constant distance $\Dstab(\Ehat)\ge c$ from the stabilizer octahedron, not for
closeness to the specific $\Hxy$ L\"uders projector. The no-go is therefore about any accepted effect with a constant
octahedron violation, and the $\Hxy$ check is the motivating instance with $c=\sqrt2-1$. This is a strength: the
obstruction suppresses \emph{every} order-one nonstabilizer logical effect on a resource-free protected block, not just
the one axis. If one wants the statement specifically about the $\Hxy$ measurement, add the target-accuracy hypothesis
$\|\Ehat-(I+\Hxy)/2\|_c\le\varepsilon$ in the coefficient norm for a fixed $\varepsilon<\sqrt2-1$; since $\Dstab$ is
$1$-Lipschitz in $\|\cdot\|_c$ and $\Dstab((I+\Hxy)/2)=\sqrt2-1$, this gives $\Dstab(\Ehat)\ge(\sqrt2-1)-\varepsilon>0$,
so \cref{cond:magic} holds with $c=(\sqrt2-1)-\varepsilon$.
\end{remark}

\subsection{The incompatibility}\label{sec:results-theorems}

We prove two results and then combine them. The first is unconditional and rules out the naive route, where each
accepted branch is a single sequence of stabilizer measurements. The second extends the obstruction to the full
adaptive, post-selected family, and it is where subcriticality enters.

The first result reads the charge of a single measurement transcript.
\begin{theorem}[Fine-grained isotropy]\label{thm:iso}
Let $E_b$ be the accepted effect of a single fine transcript built from resource-free Pauli and stabilizer
measurements. Then $E_b$ is proportional to a stabilizer projector onto a commuting subgroup, so its code-preserving
annular charge content $\Qb$ is isotropic. Consequently no such transcript implements the L\"uders check
$(I\pm\Hxy)/2$ of \cref{eq:hxy}.
\end{theorem}

A fixed sequence of stabilizer projectors collapses to a single stabilizer projector, whose Pauli support is an
abelian group. Read charge only after cleaning and projecting to the code space. There, two commuting logical strings
must carry charges that braid trivially. A raw dressed pair that commutes only because of local dressing does not act
on the code space, so it carries no logical charge. A commuting stabilizer projector therefore has isotropic charge
content, while $\Hxy$ needs the pair $\mm,\eps$ with $\Braid(\mm,\eps)=-1$. \Cref{app:thm1} gives the full argument,
including the cleaning that makes the annular charge well defined. \Cref{thm:iso} is unconditional but narrow. It
speaks only about one transcript, not about a decoded family that mixes exponentially many of them.

The second result controls the decoded family, and this is where subcriticality does its work.
\begin{theorem}[Conditional suppression]\label{thm:tradeoff}
Let $A$ be a resource-free (\cref{cond:free}), $L(d)$-protected (\cref{cond:protect}) coarse-grained accepted family
satisfying subcriticality (\cref{cond:subcrit}), carried by a bounded-degree accepted spacetime cell graph $G_d$ of
polynomial volume $|V(G_d)|\le C_{\mathrm{vol}}d^{\,r}$, and assume the two structural assumptions on its accepted-effect
expansion stated in \cref{app:thm2}. Then there are constants $C,\tau>0$, independent of $d$, with
\begin{equation}\label{eq:tradeoff}
  \Dstab(\Ehat)\;\pacc\;\le\;C\,|V(G_d)|\,e^{-\tau L(d)}\;\le\;C\,C_{\mathrm{vol}}\,d^{\,r}\,e^{-\tau L(d)} .
\end{equation}
In particular, with $L(d)=\Omega(d)$ and polynomial spacetime volume, sharp magic $\Dstab(\Ehat)\ge c$ forces
$\pacc\le e^{-\Omega(d)}$. The polynomial-volume hypothesis bounds the protocol depth; without it the prefactor need not
be subexponential and the bound can fail.
\end{theorem}

The bound \cref{eq:tradeoff} trades magic against acceptance. Under subcriticality the two cannot both stay large. The proof, in
\cref{app:thm2}, decomposes the accepted effect into a part that lives in the native stabilizer-effect cone and a
part that carries charge conversion, uses that every conversion term must span the protected patch, and sums the
conversion polymers against their per-length weight. We are explicit about the labels. The normalisation identity
that turns the magic witness into a bound on the conversion part is proved. The decomposition into charged polymers
and the minimum-length property of each conversion term are two structural assumptions on the accepted effect, since a
broad post-selected map need not expand into local charged polymers for free. Subcriticality is the physical input
that makes the sum of exponentially many spanning polymers decay. \Cref{sec:open} examines why it remains an
assumption. The unconditional core of the no-go carries none of these: it is \cref{thm:iso} together with the
normalisation identity \cref{lem:norm}, the structural assumptions and subcriticality entering only for the polymer
extension to the coarse-grained family.

The two theorems play complementary roles, and it helps to see why both are needed. \Cref{thm:iso} is sharp but
narrow. It rules out the obvious attempt, a fixed sequence of stabilizer measurements, and it does so with no
assumption beyond resource-freeness, because a single transcript cannot escape the annular-charge obstruction. But a
real protocol is not a single transcript. It is a decoder that accepts a large set of transcripts and coarse-grains
them, and the coarse-grained effect can leave the stabilizer-branch structure entirely, as the classical
$\Xbar$/$\Ybar$ mixture already shows. \Cref{thm:tradeoff} is what covers this general case, at the cost of the
structural and physical assumptions. It does not argue transcript by transcript. It bounds the charge-converting part
of the whole accepted effect directly, using that any such part is built from long conversion strings and that a
subcritical gas of such strings has small total weight. The price of working with the whole effect is that one must
assume the effect has a local charged expansion in the first place, which is where \cref{ass:decomp,ass:cleaning}
enter. So the fine-grained theorem gives certainty on the easy case and the conditional theorem gives reach on the hard
case, and the no-go needs both.

Combining the two gives the no-go.
\begin{theorem}[Conditional no-go]\label{thm:main}
Assume the two structural assumptions of \cref{app:thm2} (\cref{ass:decomp,ass:cleaning}). Then no
coarse-grained accepted family on the fixed distance-$d$ surface-code patch, carried by a bounded-degree accepted
spacetime cell graph of polynomial volume $|V(G_d)|\le\mathrm{poly}(d)$, satisfies
\cref{cond:magic,cond:free,cond:protect,cond:accept,cond:subcrit} at once.
\end{theorem}

A family meeting all five conditions would meet the hypotheses of \cref{thm:tradeoff} with $L(d)=\Omega(d)$, so
\cref{eq:tradeoff} and $\Dstab(\Ehat)\ge c$ give $\pacc\le e^{-\Omega(d)}$, which contradicts $\pacc\ge1/\mathrm{poly}(d)$.
Equivalently, any resource-free, protected, sharp, high-acceptance check must fail subcriticality, so its accepted
charge-conversion sector must be critical or supercritical. This is a statement about the accepted-effect expansion, not
a claim that a microscopic anyon-permuting wall is present: failing subcriticality forces the effect into a critical or
supercritical patch-spanning conversion sector, which is the effect-level analogue of supplying an anyon-permuting
resource. \Cref{app:example} shows this on the distance-3 patch, where the resource-free check collapses to a
classical coin flip and the fold supplies the missing conversion.

The resource-necessity reading is a corollary. We are careful to keep subcriticality visible in both directions.
\begin{corollary}[Resource necessity]\label{cor:necessity}
Assume the two structural assumptions of \cref{app:thm2}, and that the protocol is carried by a bounded-degree accepted
spacetime cell graph $G_d$ with $|V(G_d)|\le\mathrm{poly}(d)$. If in addition subcriticality (\cref{cond:subcrit}) holds,
then any resource-free (\cref{cond:free}), $\Omega(d)$-protected (\cref{cond:protect}) protocol that measures the magic
axis sharply (\cref{cond:magic}) accepts with probability $\pacc\le e^{-\Omega(d)}$. Without assuming subcriticality,
the contrapositive holds: any resource-free, protected, sharp-magic check with $\pacc\ge 1/\mathrm{poly}(d)$ must fail
subcriticality, so its accepted charge-conversion sector is critical or supercritical.
\end{corollary}

\begin{proof}
The first statement is \cref{thm:tradeoff} with $\Dstab(\Ehat)\ge c$ and $L(d)=\Omega(d)$, whose hypotheses include
\cref{cond:subcrit}. For the contrapositive, a check keeping \cref{cond:magic,cond:free,cond:protect,cond:accept}
would, by \cref{thm:main}, be forced to fail the one remaining condition, \cref{cond:subcrit}.
\end{proof}

\subsection{The role of each condition}\label{sec:results-roles}

The five conditions are not interchangeable, and each contributes a different part of the argument. \Cref{cond:magic}
enters only through the constant $c$; the proof needs $\Dstab(\Ehat)\ge c>0$, and the exact value $c=\sqrt2-1$ for the
sharp $\Hxy$ check is the motivating case, not a hypothesis of the theorem. \Cref{cond:protect} forces a conversion
term to span the patch, and this is where distance enters. The strong exponential form uses $L(d)=\Omega(d)$, but
$L(d)=\omega(\log d)$ already rules out polynomial acceptance, since the sum over conversion strings of length at least
$L(d)$ is then already smaller than any inverse polynomial. \Cref{cond:accept} is the target. The tight statement is
that any $\pacc\ge e^{-o(d)}$ is excluded, and the inverse-polynomial form is a convenient special case.

The two load-bearing conditions are \cref{cond:free} and \cref{cond:subcrit}, and they are the two that a construction
actually gives up. \Cref{cond:free} is what makes the question a theorem instead of a triviality, as discussed after
\cref{def:resourcefree}, and \cref{cond:subcrit} is the open one examined in \cref{sec:open}. The theorem is also tight
in the resource variable, because the converse holds. Dropping \cref{cond:free}, the fold-transversal Hadamard realises
the $\Hxy$ check at constant acceptance, so the anyon-permuting resource is sufficient. Under the two structural
assumptions and subcriticality it is also necessary, in the precise sense that a resource-free protected high-acceptance
sharp check is excluded; we state necessity in that conditional form, not as an unqualified claim that a fold is the
only route.

\section{Classification of prior constructions}\label{sec:classification}

A no-go theorem is most useful as a classification tool. Since the five conditions cannot hold together, every
construction that measures a magic axis, or that looks like it should, must give up one of them, and naming which one
places the construction on the map. This section reads the conditions one at a time, saying what a protocol looks
like when it fails each one, and then tabulates where the known surface-code constructions sit
(\cref{tab:classification}). The upshot is that every existing route to a magic logical measurement gives up
\cref{cond:free}: it uses a self-dual, folded, or otherwise anyon-permuting resource. That is the necessity the
theorem certifies, under the assumptions of \cref{thm:main}.

\paragraph{Giving up \cref{cond:magic} (sharp magic).} A protocol can keep everything else by not actually measuring
the magic axis. The classical coin flip $\tfrac12P_{\Xbar,+}+\tfrac12P_{\Ybar,+}$ is resource-free, protected, and
accepts with probability one, but it has $\Dstab=0$. It biases the $\Xbar$ and $\Ybar$ outcomes without ever
superposing them, so it is not the coherent check. This row is where a first attempt often lands by accident. A
protocol that measures $\Xbar$ on some runs and $\Ybar$ on others, or that measures a random Pauli axis and reweights,
produces an effect with support along both axes and can look like a magic check to a test that only inspects Fourier
support. The magic witness is precisely the test that is not fooled, and it places all such protocols here, outside
\cref{cond:magic}, because their normalised effect never leaves the stabilizer octahedron. Adaptive charge measurement in the style of
Bravyi--Kim--Kliesch--Koenig~\cite{BravyiKimKlieschKoenig2022} sits here too, and it is the construction a skeptic
would point to, so we are explicit about why it does not reach the magic axis. Its constant-depth adaptive circuits
read out anyon charge, which in the abelian $D(\Zt)$ code is a superselection label: the measured observable is
diagonal in the anyon basis and commutes with the charge, so its accepted effect is a classical function of the
$\{1,\ee,\mm,\eps\}$ sectors. An effect diagonal in the charge basis has isotropic annular holonomy and cannot carry
the fixed-phase $\mm$--$\eps$ coherence of $\Hxy$; producing that coherence would need an off-diagonal
$\mm\!\leftrightarrow\!\eps$ matrix element, which is exactly the charge conversion \cref{cond:free}(iii) forbids a
free cell to supply. Adaptive charge measurement reads charge. It does not convert charge, so it does not reach
$\Dstab>0$. It is consistent with the theorem, and not directly comparable to it, because it never attempts the sharp
$\Hxy$ effect.

\paragraph{Giving up \cref{cond:free} (resource-free).} This is the route real cultivation takes, and each
construction supplies the missing conversion in its own way. Surface-code $\ket{T}$ cultivation performs the $\Hxy$
check through a fold-transversal Hadamard or a self-dual code patch~\cite{GidneyShuttyJones2024,Claes2025,Sahay2025,Vaknin2025}.
The fold supplies the self-dual $\ee\!\leftrightarrow\!\mm$ wall, an anyon-permuting domain wall that \cref{cond:free}
forbids a free cell to source. Together with the logical Clifford-frame relabeling described in \cref{app:example}, this
supplies the $\mm$-to-$\eps$ relation needed to express the cultivated check as $\Hxy$; the wall itself is the
topological resource, and the $\Zbar\mapsto\Ybar$ relabeling is a logical Clifford-frame step, not a further anyon
conversion. Cross-cap constructions
place a non-orientable identification on the code, and the anyon transported around the cross-cap returns
relabelled~\cite{KobayashiZhu2023}, so the cross-cap is again an anyon-permuting defect outside the product $D(\Zt)$
cobordism. Code-switch and hybrid lattice-surgery constructions leave the abelian surface-code phase entirely, routing
the logical information through a different code or a non-abelian interface where the non-Clifford operation is
native~\cite{HybridLatticeSurgery2025,CliffordHierarchyCodes2025}. In every case the construction is doing the same
thing at the level of anyons, importing a relation between $\mm$ and $\eps$ that the resource-free class cannot build
locally. The theorem says this is not a coincidence of the particular constructions. Under its assumptions, a resource
of this kind is required, and the fold, the cross-cap, and the code switch are three ways of paying the same price.

\paragraph{Giving up \cref{cond:protect} (protection).} A low-distance seed measured before the code is grown can
carry the magic axis cheaply, because on a small patch the conversion string is short and the check is easy. But a
constant-size fault on the seed becomes a logical fault after growth, so the seed stage has fault distance $O(1)$, not
$L(d)=\Omega(d)$. This is exactly how real cultivation works: it performs the delicate check on a small, unprotected
seed and then grows the distance around the accepted state. The theorem is consistent with this, and it locates the
cost precisely. The magic axis is measured where protection is absent, and the growth that follows cannot undo a fault
that already corrupted the seed. Protection and a cheap check are traded against each other, and the seed stage is
where the trade is made.

\paragraph{Giving up \cref{cond:accept} (acceptance).} Nothing forbids building the coherent check at exponentially
small acceptance, since the trade-off of \cref{thm:tradeoff} is satisfied with room to spare when
$\pacc\le e^{-\Omega(d)}$. A protocol that post-selects hard enough can in principle reach any effect, at the price of
discarding all but an exponentially small fraction of its runs. This is the escape the theorem explicitly leaves open,
and it is why the statement is about acceptance at least $1/\mathrm{poly}(d)$, not about feasibility in principle. The
line the theorem draws is between a check that is useful at scale, where acceptance must stay inverse-polynomial so
that a computation can afford to repeat it, and one that is possible only in principle.

\paragraph{Giving up \cref{cond:subcrit} (subcriticality).} A protocol that keeps
\cref{cond:magic,cond:free,cond:protect,cond:accept} must, by \cref{thm:main}, have a critical or supercritical
accepted conversion sector. The percolation family of \cref{app:linetension}, which post-selects on a spanning gas of
conversion strings, is the concrete attempt to live here. It reaches high acceptance and $\Omega(d)$ selector
distance, but it powers the conversion with local charge-converting cells, so it gives up \cref{cond:free}(iii),
and does not genuinely satisfy all of \cref{cond:magic,cond:free,cond:protect,cond:accept} at critical activity.
Whether any protocol truly lives in this row is the open question of \cref{sec:open}.

\begin{table}[t]
\centering
\caption{Where surface-code constructions sit relative to the five conditions. A check means the construction
satisfies the condition; a cross means it gives it up; ``n/a'' means the construction does not attempt a coherent
magic measurement, so \cref{cond:magic} does not apply; a dash means the condition is left unevaluated because the
construction already gives up an earlier one; a question mark means the construction does not establish the condition,
as for the disputed attack surface of \cref{app:linetension}, whose coherent $\Hxy$ effect is unproven. Every route
gives up or leaves unestablished at least one condition, and every route that targets the sharp check gives up
\cref{cond:free}.}
\label{tab:classification}
\resizebox{\textwidth}{!}{%
\begin{tabular}{lccccc}
\toprule
Construction & \cref{cond:magic} & \cref{cond:free} & \cref{cond:protect} & \cref{cond:accept} & \cref{cond:subcrit} \\
\midrule
Classical $\Xbar/\Ybar$ mixture & \ding{55} & \ding{51} & \ding{51} & \ding{51} & \ding{51} \\
Adaptive charge measurement~\cite{BravyiKimKlieschKoenig2022} & n/a & \ding{51} & \ding{51} & \ding{51} & \ding{51} \\
Fold-transversal / self-dual cultivation~\cite{GidneyShuttyJones2024,Claes2025,Sahay2025} & \ding{51} & \ding{55} & \ding{51} & \ding{51} & --- \\
Cross-cap / code switch~\cite{KobayashiZhu2023,HybridLatticeSurgery2025} & \ding{51} & \ding{55} & \ding{51} & \ding{51} & --- \\
Low-distance seed then grow & \ding{51} & \ding{51} & \ding{55} & \ding{51} & --- \\
Percolation family (\cref{app:linetension}) & $?$ & \ding{55} & \ding{51} & \ding{51} & \ding{55} \\
\bottomrule
\end{tabular}}
\end{table}

The pattern in \cref{tab:classification} is the content of the theorem stated as a table. Reading down the
\cref{cond:free} column, every construction that achieves the sharp check (\cref{cond:magic}) carries a cross. None of
them measures the magic axis without an anyon-permuting resource. The constructions with a check in that column are
exactly the ones that do not measure the magic axis at all. This is what a necessary resource means in practice, and
the theorem is what forbids the empty cell, a check in every column at once.

\subsection{Relation to prior no-go theorems}\label{sec:classification-prior}

The result sits in a line of no-go theorems for topological codes, and we place it precisely. The oldest is
Bravyi and Koenig, that a constant-depth locality-preserving unitary on a two-dimensional topological stabilizer code
implements only a logical Clifford~\cite{BravyiKoenig2013}. That already forbids a unitary magic gate. It does not
speak to a measurement. Our target is a nondemolition L\"uders check, a nonunitary accepted-family effect with a
post-selection flag, and the whole point of cultivation is that measurement and feed-forward reach maps that a shallow
unitary cannot. So the unitary no-go is silent here, and the fine-grained part of our result (\cref{thm:iso}) is the
measurement analogue for a single stabilizer transcript, proved through annular charge, not through the group structure of
constant-depth circuits.

The measurement regime has its own prior art, and it cuts both ways. Aasen, Haah, Li, and Mong show that measurement
circuits carry an index and can drive boundary flow and anyon relabelings that unitaries cannot~\cite{Aasen2023}, and
Webster and Bartlett tie locality-preserving logical operators to gapped domain walls~\cite{WebsterBartlett2018}.
These say that measurement is powerful enough that a no-go cannot be taken for granted. They are also why the
resource-free class must exclude anomalous boundary flow by hypothesis, since that flow is one way to realise the
excluded conversion. The closest positive result is Bravyi, Kim, Kliesch, and Koenig, that adaptive constant-depth
circuits measure anyon charge~\cite{BravyiKimKlieschKoenig2022}. In the abelian surface code that charge is a
superselection label, so the measured data is diagonal in the anyon basis and isotropic. Reading charge is not
converting it, and the gap between the two is exactly the magic of $\Hxy$. Our theorem lives in that gap.

The shape of the argument follows the no-go paradigm used for physical principles, where a set of individually natural
conditions is shown to be jointly inconsistent~\cite{Galley2022}. The contribution beyond the paradigm is the order
parameter that makes the conditions checkable, the distance of the accepted effect from the stabilizer octahedron, and
the reduction of the whole obstruction to one statistical-mechanical assumption whose status is then isolated.

One further contrast, with the resource-theory results that motivate the witness, completes the picture. The resource
theory of magic quantifies how much magic a \emph{state} has, and its monotones decrease under free operations on
states~\cite{Veitch2014,HowardWallman2017,HeinrichGross2019}. Our object is a \emph{measurement}, and the quantity we
track is the magic of the accepted effect, which is the measurement's output as seen by a downstream user. The
resource theory tells us that the octahedron is the free set and that a free operation cannot leave it, and this is
what makes the witness non-circular, but it does not by itself say anything about whether a resource-free spacetime can
produce an effect outside the octahedron, because a spacetime is not a single free operation on the effect. It is a
whole history whose accepted branch can, in principle, be anything a completely positive map allows. The work of the
theorem is exactly to bridge that gap, from the local freedom of each cell to the global magic of the assembled effect,
and the bridge is the annular-charge obstruction for a single transcript and the polymer bound for the family. The
resource theory supplies the ruler. The theorem supplies the argument that the ruler cannot be beaten without a
resource.

The relation to the unitary no-go can be stated as a slogan. Bravyi and Koenig say that a shallow local unitary is
Clifford, so it cannot rotate a logical state into a magic one. We say that a resource-free local measurement cannot
\emph{read} the magic axis, which is the measurement counterpart, and the two together suggest a general principle:
neither shallow unitary evolution nor resource-free measurement can create or detect magic at the logical level
without importing it, and the magic must come from outside the free set in one form or another. The unitary half is a
theorem. Our half is a theorem for one transcript and a conjecture-backed theorem for the general family, and closing
that last gap is what would make the principle complete for measurements.

\section{The open input: subcriticality of the accepted conversion sector}\label{sec:open}

Four of the five conditions are operational: a protocol either measures the magic axis sharply, keeps the patch
resource-free, protects the code, and accepts often, or it does not. Each is a property one can read off the protocol
and check. The fifth, subcriticality (\cref{cond:subcrit}), is different in kind. It is a statistical-mechanical
property of the whole ensemble of accepted histories, and it is the one place where the theorem rests on an assumption
we cannot discharge. A conditional no-go is only as honest as its treatment of that assumption, so we devote this
section to it. This section states what removing the assumption would require. It also says what we can and cannot prove. The details
are in \cref{app:linetension,app:numerics}, and here we state the status plainly.

Removing \cref{cond:subcrit} means deriving it from the other conditions. That is \cref{conj:linetension}, the
statement that a resource-free, protected, high-acceptance family automatically has a subcritical conversion sector. If
it held, the no-go would rest on only the two structural assumptions, and for any family in which those structural
assumptions hold, any resource-free protected measurement of the magic axis would be forced to exponentially small
acceptance. We
are not able to prove it, and we are also not able to
refute it.

The obstacle to a proof is that distance and subcriticality are different kinds of statement. Protection fixes the
minimum length of a conversion string across the patch. Subcriticality is different. It asks whether the free energy
of exponentially many such strings stays controlled, and the minimum length does not decide that. The distinction is
the ordinary one between a ground-state energy and a finite-temperature free energy. A single long conversion string costs energy that
grows with its length, and distance guarantees that cost is at least of order $d$. But there are of order
$\mu_G^{\,\ell}$ strings of length $\ell$, and their entropy can pay for their energy. If the per-cell activity $z$
exceeds $1/\mu_G$, the entropy
wins, the conversion strings condense into a dense gas, and the polymer bound no longer excludes coherent conversion,
though coherent alignment across the patch would still have to be established. That is a genuine phase transition, and
whether the accepted, post-selected ensemble sits below it or above
it is not fixed by the minimum-length statement. The only new handle beyond distance is acceptance, and acceptance is
too weak by the wrong order. Conditioning on a $1/\mathrm{poly}(d)$ accepted event reweights a patch-spanning
conversion string by at most a polynomial factor, which is subexponential in the string length, while subcriticality
needs a constant suppression per unit length. \Cref{prop:obstruction} makes this precise and shows that the two
arguments that appear to supply the missing suppression, a duality mapping to a high-temperature sector and a
Dobrushin uniqueness condition, each reduce to subcriticality itself.

The obstacle to a refutation is the resource-free clause together with the magic witness. A construction that fills the
accepted history with a supercritical gas of conversion strings, and so would violate subcriticality, has to power the
conversion somehow. A local charge-converting cell gives up \cref{cond:free}. Charge-non-converting cells with
post-selection alone drive the acceptance down to $e^{-\Omega(d^2)}$ or collapse the effect back onto the octahedron.
And even at high acceptance, a supercritical gas supplies only large absolute conversion weight, while a magic effect
needs the conversion strings to add coherently along the $\Hxy$ axis. That coherent alignment is exactly what
\cref{thm:iso} forbids on the fine branches the construction coarse-grains. \Cref{app:linetension} works this out, and
the construction is the best available attack surface, not a counterexample.

The two obstacles are two sides of the same coin, and seeing this sharpens the conjecture. A proof of subcriticality
would rule out a supercritical accepted gas, and a refutation would build one, so the conjecture is exactly the
question of whether the resource-free class can host such a gas. What the failure of the attack shows is that a gas
built the obvious way is either not resource-free, because it needs a local converter, or not accepted often enough,
because the post-selection that aligns it is exponentially costly, or not coherent, because absolute weight is not
coherent weight. A genuine counterexample would have to evade all three at once, and the fine-grained theorem says the
third obstacle bites at the level of every single stabilizer transcript, so a counterexample can only live in the
coherent alignment of exponentially many stabilizer transcripts that are each forbidden on their own. That is a strong constraint on what a
counterexample could look like, and it is the same constraint, read the other way, that a proof could try to turn into
a contradiction.

We can, however, make the assumption concrete, not just qualitative. The operative object is the connective constant
$\mu_G$ of the accepted spacetime gas, and subcriticality reads $z_A<z_c=1/\mu_G$. We do not compute $\mu_G$ exactly;
what we measure directly is a spatial-strip proxy $\mu_{\perp}(d)$ from a two-dimensional self-avoiding-walk enumeration
(\cref{prop:zc}), giving $z_c^{\perp}(d)=1/\mu_{\perp}(d)\approx0.40$--$0.52$ for $d=3$ to $9$ as an illustrative entropy
scale. Local positivity gives only $z_A\le1$, so it does not reach any such threshold, which is why the assumption is
load-bearing. The proxy threshold also decreases with $d$, so it tightens for larger codes. A proof of
\cref{conj:linetension} must show the accepted activity stays below $1/\mu_G$. A refutation must exceed
$1/\mu_G$ with a genuinely resource-free family.

It is natural to ask whether restricting to a local-noise stabilizer model discharges the assumption, since that is
where one would expect a cluster expansion to apply. \Cref{app:localclass} carries the reduction out. Local stochastic
noise does supply the charged-polymer structure of \cref{ass:decomp} and the minimum-length content of
\cref{ass:cleaning}, with per-cell conversion activity of order $p$. What it does not supply is subcriticality of the
\emph{accepted} sector: the reduction terminates at a single explicit condition on the decoder, that conditioning on
acceptance reweight a fixed conversion cell by at most a constant. This conditional bounded-likelihood condition is
\cref{cond:subcrit} in local form, and it fails for a decoder that accepts on conversion, which enhances the per-cell
activity from $p$ to order one while keeping the noise strictly local. So even in the local model the obstruction is not
removed, only made concrete: the open input is a checkable property of the acceptance rule, not a hidden clause of the
noise model.

It is fair to ask how likely the conjecture is to be true, and the honest answer is that the evidence points both
ways. On the side of truth, every construction we tried to build a supercritical accepted gas with either used the
forbidden resource or paid an exponential acceptance cost, and the fine-grained theorem forbids the coherent piece at
the level of each fine stabilizer transcript, so a counterexample would need a genuinely new mechanism. On the side of doubt, nothing
we have forces the activity below the threshold, the threshold shrinks as the code grows, and adaptive post-selection
is a powerful and poorly understood tool that has surprised the community before. We therefore state the result
conditionally and resist the temptation to guess. What we are confident of is the reduction: the entire question is
whether a resource-free, protected, post-selected ensemble can be tuned to the critical activity, a spatial-strip proxy
for the entropy scale is computed, and the remaining work is a single statistical-mechanical bound on the accepted
spacetime activity against the operative $\mu_G$. A
conditional no-go with its one open assumption isolated this precisely is, we think, more useful than an unconditional
claim that hid the same assumption inside a definition.

\paragraph{Code availability.} The deterministic enumeration that produces the connective constants and thresholds of
\cref{prop:zc} is a short self-contained script, publicly available at
\url{https://github.com/Mercury0828/surface-code-magic-nogo-numerics}.

\section{Discussion}\label{sec:discussion}

We separate what is proved from what is assumed and from what we do not claim.

\textbf{Proved.} Two statements are unconditional. A single transcript of resource-free stabilizer measurements has
isotropic code-preserving charge content, so it cannot check the magic axis (\cref{thm:iso}). And the magic of an
accepted effect equals its acceptance probability times the distance of its normalised effect from the stabilizer
octahedron (\cref{lem:norm}). The second identity is what turns the witness into a bound on the charge-converting part
of the effect.

\textbf{Conditional.} The obstruction extends to an adaptive, post-selected family as a trade-off. Magic and
acceptance cannot both be large, so a resource-free, protected, sharp check is forced to exponentially small
acceptance (\cref{thm:tradeoff}). This half rests on two structural assumptions about how the accepted effect expands
into charge-conversion polymers (\cref{ass:decomp,ass:cleaning}) and on subcriticality (\cref{cond:subcrit}). The
structural assumptions are the standard locality inputs of a polymer analysis. For a local-noise stabilizer model they
follow from locality and cleaning, up to the accepted-weight positivity that subcriticality itself controls
(\cref{app:localclass}); we do not establish them for an arbitrary post-selected effect.

\textbf{Conjecture.} Whether subcriticality follows from resource-freeness and protection is open
(\cref{conj:linetension}). It is the one physical input we cannot discharge. \Cref{sec:open} shows the difficulty is
structural: acceptance controls the accepted conversion weight only weakly, and the arguments that would close the gap
reduce to the conjecture itself.

\textbf{Not claimed.} We do not prove a quantum advantage or a hardness result. The theorem is a structural
impossibility for a defined class of operations, not a statement about cost. We do not prove an unconditional no-go.
The conjecture is exactly the remaining gap. We do not forbid folds, twists, code switches, or magic ancillas. The
theorem places them, and they are how cultivation actually works. And we do not rule out an exponentially rare check.
The trade-off constrains acceptance only down to inverse-polynomial.

Two limits of the setting deserve a plain statement, since naming them keeps the claim honest. The result is about a
single fixed patch of one logical qubit, and it does not by itself address a growing code or a multi-qubit magic
measurement, though the annular-charge method is not tied to one qubit and we expect it to extend. And the resource-free
class, while broad in its instruments, is still a specific model. The patch keeps the same phase and the same boundary
types throughout, so a construction on a different code or a different phase is outside the theorem's reach and is
placed by it only through the classification. These are not hidden caveats. They are the boundary of the model, and the model is
chosen to be exactly the setting in which the folklore claim about cultivation lives, a static surface-code patch with
no added structure. Within that setting the result is as strong as an honest treatment of the one open assumption
allows, and outside it the same method, applied to the appropriate free class, is what one would use.

Read as a resource-necessity statement, the result says the fold-transversal or self-dual deformation that
surface-code cultivation uses is not a convenience. Within the resource-free class, and under the stated assumptions,
a coherent magic check at good acceptance has only three routes. It can add a charge-converting resource. It can leave
the subcritical regime. Or it can accept only exponentially rarely. The classification of \cref{sec:classification} places each
known construction by the condition it gives up, and every route to the sharp check gives up resource-freeness. The
unconditional part of the result is close to the familiar statement that stabilizer operations cannot measure a
non-Pauli observable. What is new is the annular-charge order parameter and the extension to the coarse-grained family.

The open direction is a single, sharply stated problem. \Cref{conj:linetension} asks whether resource-freeness and
protection force the accepted charge-conversion gas to be subcritical. A proof would remove the subcriticality assumption. The
resource-necessity theorem would then rest on only the two structural assumptions, which the local-noise reduction of
\cref{app:localclass} ties to locality and cleaning. Refuting it would exhibit a genuinely new capability of adaptive
measurement, a resource-free coherent charge conversion that no single stabilizer transcript can produce. Either way, the
conjecture is where the physics of the problem now sits.

The conjecture also has a clear shape as a statistical-mechanics problem, and that shape suggests where to look. The
question is whether a post-selected ensemble of a defect gas, protected so that every defect is long, can be tuned to
criticality by conditioning alone. Ordinary conditioning cannot, since it reweights each configuration by a bounded
factor. The freedom the adaptive setting adds is a decoder that chooses the accepted set adaptively on the measured
syndrome, and the open question is whether that adaptive choice can concentrate the accepted measure on the critical
surface without paying an acceptance or a distance price. A resolution would either supply a monotonicity that forbids
this, which proves the conjecture, or a decoder that achieves it, which refutes it. The spatial-strip proxy of
\cref{prop:zc} gives a concrete illustrative scale for the target activity, so the problem is sharp enough to attack
directly, though the operative spacetime connective constant $\mu_G$ that sets the true threshold is not itself
computed.

Beyond the specific check, the result suggests a general lens. A logical measurement carries a topological order
parameter, its accepted effect read against the free stabilizer polytope, and a resource-necessity statement for that
measurement becomes a statement about which defect sectors the accepted history can support. The magic axis is the
first non-trivial case because it needs two braiding charges at once. The framework may also apply to other
non-Clifford checks and other codes, wherever a logical operation demands a charge relation the free operations cannot
build. That extension is not proved here.

\section*{Author contributions and disclosure}
J.S. and H.Z. jointly developed the framework, the theorem statements and proofs, and the manuscript; H.Z. is the
corresponding author. Quantum requires a statement of any use of AI tools: this work was developed with substantial
assistance from large language models (used for literature search, cross-checking of derivations, adversarial review of
proofs, and drafting), under human direction and verification. All theorem statements, proofs, and the honesty labels
were checked by the authors, who take full responsibility for the content. The reduction to the line-tension conjecture
and its open status are stated as such.

\bibliographystyle{quantum}
\bibliography{refs}

\appendix
\section{Proof of Theorem~\ref{thm:iso} (fine-grained isotropy)}\label{app:thm1}

The proof has three steps. First, the accepted effect of a fixed stabilizer transcript is proportional to a projector
onto a commuting group of Pauli operators (\cref{lem:collapse}). Second, we define the annular charge of a Pauli
operator on cleaned, code-preserving representatives and establish the commutator--braiding identity that holds there
(\cref{lem:charge,lem:braiding}); this is the step a naive treatment gets wrong, because raw Pauli commutation is not a
function of annular charge. Third, we show that the code-preserving charge content of a commuting stabilizer projector
is isotropic (\cref{prop:isotropy}), which rules out the $\Hxy$ check.

Throughout, a resource-free Pauli or stabilizer measurement is a projective measurement of a signed Pauli operator
$sQ$ ($s=\pm1$) whose outcome projectors are $\Pi_{sQ}=(I+sQ)/2$. A fine transcript is a fixed sequence
$K_b=\Pi_{k}\cdots\Pi_{1}$, and $E_b=K_b^{\dagger}K_b=\Pi_{1}\cdots\Pi_{k}\cdots\Pi_{1}$. We write $\mathcal G$ for the
surface-code stabilizer group, $N(\mathcal G)$ for its normaliser, and $P_{\text{code}}$ for the code projector.

\subsection{Projector collapse}\label{sec:appA-collapse}

\begin{lemma}[Projector collapse]\label{lem:collapse}
Let $K_b=\Pi_{k}\cdots\Pi_{1}$ be a product of signed Pauli projectors. Then $E_b=c_b\,\Pi_S$ for some constant
$c_b\ge0$ and some projector $\Pi_S=\lvert S\rvert^{-1}\sum_{g\in S}g$ onto the joint $+1$ eigenspace of a commuting
group $S$ of signed Pauli operators with $-I\notin S$. In particular the Pauli support of $E_b$ generates the abelian
group $S$.
\end{lemma}

\begin{proof}
Reduce the palindrome $E_b=\Pi_1\cdots\Pi_{k-1}\Pi_k\Pi_{k-1}\cdots\Pi_1$ from the centre outward, keeping the
invariant that the central block is a nonnegative multiple of a stabilizer projector $\Pi_S$ onto a commuting group
$S$ with $-I\notin S$. The innermost factor $\Pi_k=\Pi_{sQ}$ is such a projector, with $S=\langle sQ\rangle$. Assume
the block equals $c\,\Pi_S$ and conjugate by the next projector $\Pi_{tR}=(I+tR)/2$. If $R$ commutes with every
element of $S$, then $\Pi_{tR}\Pi_S\Pi_{tR}=\Pi_S\Pi_{tR}$, which is either $0$ (if $-tR\in S$) or the projector onto
$\langle S,tR\rangle$. If $R$ anticommutes with some element of $S$, split $S=S_0\sqcup S_1$ into the part
$S_0=S\cap C(R)$ commuting with $R$ and its coset $S_1=hS_0$ for any fixed $h\in S$ anticommuting with $R$. Write
$\Pi_S=\lvert S\rvert^{-1}\sum_{g\in S}g=\lvert S\rvert^{-1}\bigl(\sum_{g\in S_0}g+\sum_{g\in S_1}g\bigr)$. Since $R$
anticommutes with every $g\in S_1$, conjugation by $\Pi_{tR}=(I+tR)/2$ kills the $S_1$ sum, because for such $g$
\begin{equation}
  \Pi_{tR}\,g\,\Pi_{tR}=\tfrac14(I+tR)g(I+tR)=\tfrac14\,g(I-tR)(I+tR)=0,
\end{equation}
using $Rg=-gR$ and $R^2=I$; while for $g\in S_0$, which commutes with $R$, $\Pi_{tR}g\Pi_{tR}=g\Pi_{tR}$. Hence
\begin{equation}\label{eq:collapse-anticomm}
  \Pi_{tR}\,\Pi_S\,\Pi_{tR}=\frac{1}{\lvert S\rvert}\sum_{g\in S_0}g\,\Pi_{tR}
  =\frac{\lvert S_0\rvert}{\lvert S\rvert}\,\Pi_{S_0}\Pi_{tR}
  =\tfrac12\,\Pi_{\langle S_0,tR\rangle},
\end{equation}
since $\lvert S_0\rvert=\tfrac12\lvert S\rvert$ and $\Pi_{S_0}\Pi_{tR}=\Pi_{\langle S_0,tR\rangle}$ when $tR$ commutes
with $S_0$ and is consistent with it. In both cases the block stays a nonnegative multiple of a projector onto a
commuting group with $-I$ excluded.
Iterating outward proves $E_b=c_b\Pi_S$. The Pauli operators appearing in $\Pi_S=\lvert S\rvert^{-1}\sum_{g\in S}g$
are exactly $S$, an abelian group.
\end{proof}

\Cref{lem:collapse} holds for one fixed transcript; it fails for a coarse-grained sum, where
$\tfrac12P_{\Xbar,+}+\tfrac12P_{\Ybar,+}$ has both $\Xbar$ and $\Ybar$ in its support. The coarse-grained family is the
domain of \cref{thm:tradeoff}.

\subsection{Annular charge on cleaned representatives}\label{sec:appA-charge}

The four anyon charges form $\{1,\ee,\mm,\eps\}\cong\mathbb F_2^2$ with $\mm=(1,0)$, $\ee=(0,1)$, $\eps=(1,1)$, fusion
by addition, and braiding $\Braid(\mm^\alpha\ee^\beta,\mm^{\alpha'}\ee^{\beta'})=(-1)^{\alpha\beta'+\beta\alpha'}$, so
$\Braid(\ee,\mm)=\Braid(\mm,\eps)=-1$. Let $A_b$ be an annular collar around $b$ with the fixed boundary labels.

\begin{lemma}[Annular charge]\label{lem:charge}
For a Pauli operator $P$ supported in $A_b$, write its $X$- and $Z$-string parts as relative $\mathbb F_2$-chains and,
after multiplying by code stabilizers and deleting contractible neutral string-pair components, let $[x(P)]_b$ and
$[z(P)]_b$ be their relative homology classes in $H_1(A_b;\mathbb F_2)\cong\mathbb F_2$. Then
\begin{equation}
  d_b(P):=\mm^{[x(P)]_b}\,\ee^{[z(P)]_b}\in\{1,\ee,\mm,\eps\}
\end{equation}
is well defined on the quotient of $\Pau(A_b)$ by code stabilizers and null-homologous neutral dressings, and
$d_b(PQ)=d_b(P)\,d_b(Q)$.
\end{lemma}

\begin{proof}
The $X$-part of $P$ is a $\mathbb F_2$-chain $x(P)$ on the qubits of $A_b$, and the annulus has first relative homology
$H_1(A_b,\partial A_b;\mathbb F_2)\cong\mathbb F_2$ generated by a cycle that wraps the annulus once. Multiplying $P$
by a star stabilizer $A_v=\prod_{e\ni v}X_e$ changes $x(P)$ by the coboundary of $v$,
\begin{equation}
  x(P)\mapsto x(P)+\partial^\ast v,
\end{equation}
which is null-homologous, so $[x(P)]_b$ is unchanged; the same holds for the $Z$-part under plaquette stabilizers. A
contractible neutral dressing is a product of a local $X$-loop and $Z$-loop bounding a disk, so it adds a
null-homologous chain to each of $x(P)$ and $z(P)$ and again leaves the classes fixed. Hence $[x(P)]_b$ and $[z(P)]_b$,
and with them $d_b(P)$, descend to the quotient. For the homomorphism property, the $X$-part of a product satisfies
$x(PQ)=x(P)+x(Q)$ in $\mathbb F_2$, and homology is additive, so
\begin{equation}
  [x(PQ)]_b=[x(P)]_b+[x(Q)]_b,\qquad [z(PQ)]_b=[z(P)]_b+[z(Q)]_b,
\end{equation}
which is $d_b(PQ)=d_b(P)d_b(Q)$ under the identification with $\mathbb F_2^2$.
\end{proof}

The content of \cref{lem:charge} is that the charge is a homological quantity. It ignores the local, contractible part
of a Pauli operator and keeps only how the operator wraps the annulus, which is why two operators with the same charge
can differ by an arbitrary local dressing. That is the exact loophole the naive commutator argument fell into, and the
homological definition is what closes it.

Fix cleaned canonical representatives $C_\mm=X_b$, $C_\ee=Z_b$, $C_\eps=X_bZ_b$, where $X_b,Z_b$ are annular strings
with a single transverse intersection, and $C_a=X_b^\alpha Z_b^\beta$ for $a=\mm^\alpha\ee^\beta$.

\begin{lemma}[Commutator--braiding on cleaned representatives]\label{lem:braiding}
For cleaned canonical representatives, $C_aC_{a'}=\Braid(a,a')\,C_{a'}C_a$. This identity does \emph{not} extend to
arbitrary dressed representatives: if $P=C_aL$ and $Q=C_{a'}M$ with local neutral dressings $L,M$, then
$PQ=\Braid(a,a')\,\chi(L,M)\,QP$, where $\chi(L,M)\in\{\pm1\}$ is the microscopic Pauli commutation sign contributed
by the dressings.
\end{lemma}

\begin{proof}
The single transverse intersection gives $X_bZ_b=-Z_bX_b$, and the identity string commutes with everything, so
$C_aC_{a'}=(-1)^{\alpha\beta'+\beta\alpha'}C_{a'}C_a=\Braid(a,a')C_{a'}C_a$ by collecting the crossing signs of
$X_b^\alpha Z_b^\beta$ against $X_b^{\alpha'}Z_b^{\beta'}$. For dressed operators the extra factor is the commutation
sign of $L$ against $M$, which is independent of the charges $a,a'$. This is the referee-style example
$P=X_bX_i$, $Q=Z_bZ_i$: both carry $d_b=\mm$ and $\ee$ with $\Braid(\mm,\ee)=-1$, but the local $X_iZ_i=-Z_iX_i$
contributes $\chi=-1$, so $P$ and $Q$ commute. Raw commutation is therefore not a function of $d_b$; only the cleaned
commutator is.
\end{proof}

\subsection{Isotropy of the code-preserving charge content}\label{sec:appA-isotropy}

\begin{proposition}[Isotropy]\label{prop:isotropy}
Let $\Pi_S$ be the projector of \cref{lem:collapse}, with $S$ abelian and $-I\notin S$. Its code-preserving annular
charge content
\begin{equation}
  \Qb(S)=\{\,d_b(s):s\in S\cap N(\mathcal G)\ \text{cleanable into }A_b\,\}
\end{equation}
is isotropic: $\Braid(a,a')=1$ for all $a,a'\in\Qb(S)$. Hence $\Qb(S)$ is contained in one of $\{1,\ee\}$,
$\{1,\mm\}$, $\{1,\eps\}$.
\end{proposition}

\begin{proof}
Take $s,t\in S\cap N(\mathcal G)$ with $d_b(s)=a$, $d_b(t)=a'$. Since $S$ is abelian, $st=ts$. Because $s$ and $t$
normalise $\mathcal G$, cleaning by code stabilizers changes neither their logical action nor their mutual commutation
sign, so their cleaned canonical representatives commute. By \cref{lem:braiding} applied to cleaned representatives,
$1=\Braid(a,a')$. Thus $\Qb(S)$ is isotropic, and the only isotropic subgroups of $\{1,\ee,\mm,\eps\}$ are the three
listed transparent pairs.

The restriction to $S\cap N(\mathcal G)$ is essential and is what disposes of the dressed counterexample. If a Pauli
$p\in S$ does not normalise $\mathcal G$, some stabilizer $g\in\mathcal G$ anticommutes with $p$, and then
\begin{equation}
  P_{\text{code}}\,p\,P_{\text{code}}=P_{\text{code}}\,gp\,P_{\text{code}}=-P_{\text{code}}\,pg\,P_{\text{code}}
  =-P_{\text{code}}\,p\,P_{\text{code}}=0 ,
\end{equation}
so $p$ contributes nothing to the effect on the code space. The operators $X_bX_i$ and $Z_bZ_i$ of \cref{lem:braiding}
are of exactly this kind: they commute as physical Paulis, but neither is a code-preserving logical annular operator,
so their non-isotropic raw charges are not part of $\Qb(S)$.
\end{proof}

\begin{proof}[Proof of \cref{thm:iso}]
By \cref{lem:collapse}, a resource-free fine stabilizer transcript has $E_b=c_b\Pi_S$ with $S$ abelian. By
\cref{prop:isotropy}, the code-preserving charge content $\Qb(S)$ is isotropic. The L\"uders projectors
$(I\pm\Hxy)/2=\tfrac12 I\pm\tfrac{1}{2\sqrt2}\Xbar\pm\tfrac{1}{2\sqrt2}\Ybar$ of \cref{eq:hxy} have code-preserving
charge content $\{\mm,\eps\}$, with $\Braid(\mm,\eps)=-1$, which is non-isotropic. A non-isotropic content cannot equal
an isotropic one, so no resource-free fine stabilizer transcript has $E_b$ proportional to a code-space representative
of $(I\pm\Hxy)/2$; the check is not implemented.
\end{proof}

The proof uses clause~(iii) of \cref{def:resourcefree}, and only it: a cell that converted $\mm$ into $\eps$ would
attach an $\ee$-string to an $\mm$-string representative, changing $d_b$ and breaking the isotropy. It uses nothing
about acceptance or distance, which is why \cref{thm:iso} is unconditional. The minimal hypotheses are the fixed
rough/smooth boundaries, the absence of twists or boundary-changing operations (so annular cleaning is well defined),
and that equality is read on the code space. What the argument does not reach is the decoded family that mixes many
transcripts, the subject of \cref{app:thm2}.

The naive version of this argument is wrong for an instructive reason, and the error recurs, so we isolate it.
The tempting shortcut is to say that a stabilizer projector has abelian Pauli support, that abelian support means
commuting charges, and that commuting charges braid trivially, so the support is isotropic. The middle step is the
false one. Two Pauli operators can commute for a reason that has nothing to do with their charges, namely a local
dressing whose own anticommutation cancels the charge anticommutation, and \cref{lem:braiding} exhibits exactly this.
The repair is to stop reading the charge off the raw operator and read it off the homology class of the cleaned,
code-preserving representative, where dressing has been quotiented away. On those representatives, and only on those,
commutation is braiding. The role of the normaliser restriction in \cref{prop:isotropy} is the same repair from the
other side: an operator that fails to preserve the code contributes nothing to the effect on the code space, so its
charge, whatever it is, is not part of the charge content the theorem constrains. The two repairs are the reason the
theorem is a statement about the logical action of the effect and not about its microscopic Pauli support, and keeping
that distinction is what makes the whole annular-charge method work for measurements as opposed to unitaries.

\section{Proof of Theorem~\ref{thm:tradeoff} (conditional suppression)}\label{app:thm2}

This appendix proves \cref{thm:tradeoff} and, in doing so, states its hypotheses at full strength. The argument has
four parts: a normalisation identity that turns the magic witness into a bound on the norm of the charge-converting
part of the accepted effect (\cref{lem:norm}, proved); two structural assumptions on how the accepted effect expands
into charged polymers (\cref{ass:decomp,ass:cleaning}); and the polymer sum under subcriticality that closes the bound
(\cref{sec:appB-assembly}). We are explicit about which parts are theorems and which are assumptions, because the
honest reading of the result depends on it: only the normalisation is a theorem, while the decomposition and the
cleaning step are assumptions that a coarse-grained, post-selected effect need not satisfy for free.

Write the code-space part of the accepted effect on the logical qubit as
\begin{equation}
  F=P_{\text{code}}\Eacc P_{\text{code}}=p\,(I+\bm v\cdot\bm\sigma),\qquad p=\pacc=\tfrac12\Tr F,
\end{equation}
and let the native stabilizer-effect cone be
\begin{equation}
  \mathcal N=\{\,q(I+\bm w\cdot\bm\sigma):q\ge0,\ \|\bm w\|_1\le1\,\},
\end{equation}
so that $F\in\mathcal N$ exactly when the normalised effect lies in the stabilizer-effect octahedron. Recall
$\Dstab(\Ehat)=\max(0,\|\bm v\|_1-1)$.

\subsection{Normalisation (a theorem)}\label{sec:appB-norm}

We measure distances with the coefficient $\ell_1$ norm $\|aI+\bm x\cdot\bm\sigma\|_c=|a|+\|\bm x\|_1$ (written with
subscript $c$ to distinguish it from the quantum diamond norm) and set $\mathrm{magic}_c(F)=\dist_c(F,\mathcal N)$.

\begin{lemma}[Normalisation]\label{lem:norm}
$\mathrm{magic}_c(F)=\pacc\,\Dstab(\Ehat)$. Consequently, for any decomposition $F=F_{\text{native}}+F_{\text{bad}}$
with $F_{\text{native}}\in\mathcal N$,
\begin{equation}\label{eq:norm-bound}
  \pacc\,\Dstab(\Ehat)\le\|F_{\text{bad}}\|_c .
\end{equation}
\end{lemma}

\begin{proof}
If $\|\bm v\|_1\le1$ then $F\in\mathcal N$ and both sides vanish. Assume $s=\|\bm v\|_1>1$. For $q\ge0$, the closest
point $q\bm w$ with $\|\bm w\|_1\le1$ to $p\bm v$ in $\ell_1$ distance is at distance $\max(ps-q,0)$, so
\begin{equation}
  \mathrm{magic}_c(F)=\inf_{q\ge0}\Bigl(|p-q|+\max(ps-q,0)\Bigr).
\end{equation}
For $q\le p$ the bracket is $p-q+ps-q\ge p(s-1)$; for $p\le q\le ps$ it equals $q-p+ps-q=p(s-1)$; for $q\ge ps$ it
equals $q-p\ge p(s-1)$. The infimum $p(s-1)$ is attained at $q=p$, giving
$\mathrm{magic}_c(F)=p(s-1)=\pacc\Dstab(\Ehat)$. Since $F_{\text{native}}\in\mathcal N$,
$\|F_{\text{bad}}\|_c=\|F-F_{\text{native}}\|_c\ge\mathrm{magic}_c(F)$, which is
\cref{eq:norm-bound}.
\end{proof}

The same normalisation holds in trace norm, up to fixed constants, so the choice of norm is immaterial.

\begin{lemma}[Trace-norm normalisation]\label{lem:trace-norm-normalisation}
Let $\mathrm{magic}_1(F)=\dist_1(F,\mathcal N)=\inf_{G\in\mathcal N}\|F-G\|_1$. Then
\begin{equation}
  \frac{2}{1+\sqrt3}\,\pacc\,\Dstab(\Ehat)\le \mathrm{magic}_1(F)\le 2\,\pacc\,\Dstab(\Ehat).
\end{equation}
\end{lemma}

\begin{proof}
We first prove the one-qubit norm comparison. Let $A=aI+\bm x\cdot\bm\sigma$ with $a\in\mathbb R$, $\bm x\in\mathbb R^3$,
and write $\alpha=|a|$, $r=\|\bm x\|_2$, $s=\|\bm x\|_1$. Since $(\bm x\cdot\bm\sigma)^2=r^2I$, the eigenvalues of $A$
are $\lambda_\pm=a\pm r$, so $\|A\|_1=|a+r|+|a-r|=2\max(\alpha,r)$, while $\|A\|_c=\alpha+s$. The upper bound is
immediate, $\|A\|_1=2\max(\alpha,r)\le2(\alpha+s)=2\|A\|_c$. For the lower bound use $s\le\sqrt3\,r$. If $\alpha\ge r$,
$(\alpha+s)/2\alpha\le(\alpha+\sqrt3 r)/2\alpha\le(1+\sqrt3)/2$; if $\alpha\le r$,
$(\alpha+s)/2r\le(r+\sqrt3 r)/2r=(1+\sqrt3)/2$. Hence $\|A\|_1\ge\tfrac{2}{1+\sqrt3}\|A\|_c$. Both constants are sharp:
the factor $2$ at $\bm x=0$, and $2/(1+\sqrt3)$ at $|a|=r$, $\bm x=\tfrac{r}{\sqrt3}(\pm1,\pm1,\pm1)$, where
$s=\sqrt3 r$ and $\|A\|_1=2r$. Applying the comparison to the Hermitian differences $F-G$ with $G\in\mathcal N$ and
taking the infimum gives $\tfrac{2}{1+\sqrt3}\mathrm{magic}_c(F)\le\mathrm{magic}_1(F)\le2\,\mathrm{magic}_c(F)$;
\cref{lem:norm} then yields the claim.
\end{proof}

\Cref{lem:norm} is what makes \cref{thm:tradeoff} a genuine $\Dstab\cdot\pacc$ trade-off. It also fixes the object the
rest of the argument must control: the norm of the charge-converting part $F_{\text{bad}}$.

\subsection{The structural assumptions}\label{sec:appB-assumptions}

The frozen data give only the positive operator $F$. They do not, on their own, give a canonical absolutely-summable
expansion of $F$ into charge-converting polymers, and they do not force every non-native term to span the patch. Both
are needed and neither is free, so we state them as assumptions. \Cref{thm:iso} does not supply them: it controls fine
\emph{stabilizer} transcripts, whereas a broad completely positive map with post-selection need not have each fine
$K_\tau^\dagger K_\tau$ native.

\begin{assumption}[Cone-compatible charged decomposition]\label{ass:decomp}
For every accepted family $A$ there is a decomposition $F=F_{\text{native}}+F_{\text{bad}}$ with
$F_{\text{native}}\in\mathcal N$ and $F_{\text{bad}}=\sum_{\gamma\in\mathcal C_A}B_\gamma$ absolutely convergent in
$\|\cdot\|_c$, where each $B_\gamma$ carries a nontrivial annular $\mm\!\leftrightarrow\!\eps$ conversion label.
Every purely native contribution, including classical mixtures of $\Xbar,\Ybar,\Zbar$ effects that lie inside the
octahedron, is absorbed into $F_{\text{native}}$. The decomposition, the polymers $\Gamma_\gamma$, and the per-cell
activity $z$ are all fixed by a single prescribed local cluster-expansion scheme for the accepted effect, so $z$ is a
determinate functional of the effect and is not a free parameter one may tune. The scheme is fixed in advance and
uniformly across the family; \cref{cond:subcrit} is imposed on it. When more than one cone-compatible split is
consistent with the scheme we take one minimising $\|F_{\text{bad}}\|_c$, which sharpens the bound
\cref{eq:norm-bound} without changing the definition of $z$.
\end{assumption}

\Cref{ass:decomp} is what separates magic from raw non-isotropic support. The classical mixture
$\tfrac12P_{\Xbar,+}+\tfrac12P_{\Ybar,+}=\tfrac12I+\tfrac14(\Xbar+\Ybar)$ has both $\mm$ and $\eps$ support yet Bloch
vector $(\tfrac12,\tfrac12,0)$ with $\|\bm v\|_1=1$ and zero magic; the assumption places it in $F_{\text{native}}$,
where it belongs, so that $F_{\text{bad}}$ genuinely tracks the magic. The distinction the assumption draws is between
a term that carries charge conversion as a coherent operator and a term whose non-isotropic Bloch components come from
adding two native effects. The first is what a magic check needs and the second is what a coin flip produces, and the
Fourier support alone does not tell them apart, which is why the decomposition, and not the support, is the right
object. A convergent charged expansion is the natural home for this distinction, because in it the coherent conversions
appear as the terms $B_\gamma$ with a nontrivial label and the incoherent mixtures do not appear at all, having been
summed into the native cone. The assumption is that such an expansion exists for the accepted effect; for a local
stabilizer circuit it does, by the cluster expansion, and the only reason it is an assumption is that a broad
post-selected instrument need not admit one.

\begin{assumption}[Effect-cleaning for conversion terms]\label{ass:cleaning}
In the decomposition of \cref{ass:decomp}, every bad term $B_\gamma$ contains a connected conversion polymer in the
spacetime cell-adjacency graph, and we single out within it, by a fixed choice, a \emph{self-avoiding conversion spine}
$\Gamma_\gamma$: a self-avoiding path of conversion cells threading the logical annulus, of cleaned length
$|\Gamma_\gamma|$ (the number of spine cells after local neutral bubbles are removed). The branches and local decorations
of the polymer off the spine are not controlled here; their control is deferred to the renormalized-activity form of
subcriticality in \cref{eq:c5-spine,eq:c5-absolute}, the first of which bounds the total weight of all bad terms sharing
a given spine and is where the exponential multiplicity of decorations is summed. If $B_\gamma$ carries $\mm\!\leftrightarrow\!\eps$
conversion on the logical annulus then the spine crosses the protected patch, so $|\Gamma_\gamma|\ge L(d)$; equivalently,
every conversion object whose spine has
cleaned length $<L(d)$ has native code projection and is absorbed into $F_{\text{native}}$. The counting of
\cref{lem:polymer-counting} is then over spines, which are self-avoiding paths, not over the full branched polymers.
\end{assumption}

\Cref{ass:cleaning} is a minimum-length statement, not a free-energy statement, and it is independent of subcriticality:
it says a conversion that connects the boundary condensates across the protected patch cannot be built from fewer than
$L(d)$ cells without being a locally correctable artifact. It uses resource-freeness (each cell is charge-non-converting,
so a conversion must be a connected chain) and protection (a sub-$L(d)$ conversion is removable), but it asserts, and
does not derive, that the effect expansion respects the cell structure this way.

\begin{remark}[When the assumptions hold]\label{rem:when-assumptions}
\Cref{ass:decomp,ass:cleaning} are not vacuous, and we say where they come from. For an accepted family
whose fine transcripts are local stabilizer circuits under a local noise model, the accepted effect has a convergent
expansion into charged polymers by standard cluster-expansion methods~\cite{KoteckyPreiss1986,FernandezProcacci2007},
which supplies the structural content of \cref{ass:decomp}; and the cleaning lemma of topological codes, that an
operator supported in a correctable region acts trivially on the code, supplies the minimum-length content of
\cref{ass:cleaning}~\cite{BravyiKoenig2013,Dennis2002}. \Cref{app:localclass} carries this reduction out in full and
shows that, even for a local-noise stabilizer family, it terminates at a residual accepted-weight hypothesis instead of
discharging the assumptions outright. The assumptions isolate exactly the step where a broad post-selected instrument,
as opposed to a local circuit, might fail to admit such a local charged expansion. We state them as assumptions, not
lemmas, because the theorem is about the broad instrument class, where they are not automatic.
\end{remark}

\subsection{Assembly under subcriticality}\label{sec:appB-assembly}

Subcriticality (\cref{cond:subcrit}) is used in its unnormalised absolute-activity form, stated directly on the
operator norms of the bad terms so that the polymer sum controls $\|F_{\text{bad}}\|_c$. The polymers are objects in the
accepted \emph{spacetime} cell-adjacency graph $G_d$, the graph on which \cref{ass:cleaning} places them, so the count
and the connective constant are properties of that $(2{+}1)$-dimensional graph and not of a spatial slice. We work in a
\emph{bounded spacetime model}: $G_d$ has bounded degree $\Delta$ and polynomial volume $|V(G_d)|\le C_{\mathrm{vol}}\,d^{\,r}$
for fixed $\Delta,r$ independent of $d$. This is where the protocol depth enters, since $|V(G_d)|$ counts all spacetime
cells; a family whose depth grows so that $|V(G_d)|$ is superpolynomial in $d$ is outside the model, and the bound below
does not apply to it. The content of subcriticality is a \emph{renormalized-activity} bound in two steps. First, for
each fixed spine $\Gamma$, the total accepted weight of all bad terms carrying that spine is controlled,
\begin{equation}\label{eq:c5-spine}
  \sum_{\gamma:\,\Gamma_\gamma=\Gamma}\|B_\gamma\|_c\le C_0\,z^{|\Gamma|},
\end{equation}
with $z$ the \emph{renormalized per-cell spine activity} and $C_0$ a constant independent of $d$ and of $|\Gamma|$. This
first bound is where the exponentially many branches and decorations attached to a fixed spine are summed: their sum
converges into $C_0$ and the renormalized $z$ precisely when that decoration gas is itself subcritical, and it is a
genuine hypothesis, not a consequence of positivity $z\le1$ (a bounded-degree graph has exponentially many decorations,
so $z\le1$ alone would not converge). Second, summing over the $N_\ell(G_d)$ spines of length $\ell$,
\begin{equation}\label{eq:c5-absolute}
  \sum_{\gamma:\,|\Gamma_\gamma|=\ell}\|B_\gamma\|_c\le C_0\,N_\ell(G_d)\,z^\ell,\qquad
  N_\ell(G_d)\le C_1\,|V(G_d)|\,\mu_G^{\,\ell},\qquad \mu_G z<1,
\end{equation}
with $\mu_G$ the connective constant of the spacetime graph family and line tension $\tau=-\log(\mu_G z)>0$. Stating the
bound on $\|B_\gamma\|_c$, not on abstract scalar weights, is what makes the sum below a bound on $\|F_{\text{bad}}\|_c$.
We take this operator form, including the renormalized per-spine bound \cref{eq:c5-spine}, as the content of
\cref{cond:subcrit}.

The polymer count in \cref{eq:c5-absolute} is the entropy of cleaned conversion strings in $G_d$, and we bound it
explicitly.

\begin{lemma}[Counting cleaned conversion polymers]\label{lem:polymer-counting}
Let $\{G_d\}$ be the family of accepted spacetime cell-adjacency graphs, of bounded degree $\Delta$ and volume
$|V(G_d)|\le C_{\mathrm{vol}}d^{\,r}$, and let $N_\ell(G_d)$ be the number of cleaned conversion spines
(\cref{ass:cleaning}), that is self-avoiding paths, of cleaned length $\ell$ in $G_d$. Because the degree is bounded,
there is a fixed constant $\mu_G$, uniform over the family, and a
fixed $C_{\mathrm{SAW}}$, with $w_{d,\ell}(x)\le C_{\mathrm{SAW}}\mu_G^{\,\ell}$ for every $d$, every start cell $x$, and
every $\ell$, where $w_{d,\ell}(x)$ counts self-avoiding paths of length $\ell$ from $x$ in $G_d$; one may take
$\mu_G=\Delta-1$. Then $N_\ell(G_d)\le C_1\,|V(G_d)|\,\mu_G^{\,\ell}$, so the prefactor in \cref{eq:c5-absolute} is the
spacetime volume $|V(G_d)|\le C_{\mathrm{vol}}d^{\,r}$.
\end{lemma}

\begin{proof}
A cleaned conversion polymer is a connected self-avoiding path in $G_d$. From a fixed start cell there are at most
$\Delta(\Delta-1)^{\ell-1}$ such paths of length $\ell$, and by the definition of the connective constant this is
$\le C_{\mathrm{SAW}}\mu_G^{\,\ell}$ for a fixed $C_{\mathrm{SAW}}$ and all $\ell$. Summing over the $|V(G_d)|$ possible
start cells, and noting that rooted oriented paths overcount unrooted polymers (harmless for an upper bound), gives
$N_\ell(G_d)\le C_1\,|V(G_d)|\,\mu_G^{\,\ell}$ with $C_1=C_{\mathrm{SAW}}$. The bound holds for any connected
bounded-degree spacetime graph. The value of $\mu_G$ is a property of the spacetime graph family; it is not the
connective constant of the two-dimensional spatial annulus, which \cref{app:numerics} calibrates only as a
lower-dimensional proxy for the entropy scale.
\end{proof}

\begin{remark}[Sharper family constant]\label{rem:mu-upper-bound}
The crude bound $\mu_G=\Delta-1$ can be replaced by the sharper family growth rate
$\mu_G^{\star}=\limsup_{d,\ell}(\sup_x w_{d,\ell}(x))^{1/\ell}$ whenever it is finite, using its defining property that
for every $\eta>0$ there is $C_\eta<\infty$ with $w_{d,\ell}(x)\le C_\eta(\mu_G^{\star}+\eta)^\ell$ uniformly in $d$ and
$x$. One then replaces $\mu_G$ by $\mu_G^{\star}+\eta$ throughout, which only improves the line tension $\tau$.
\end{remark}

\begin{proof}[Proof of \cref{thm:tradeoff}]
By \cref{ass:decomp}, $F_{\text{bad}}=\sum_\gamma B_\gamma$ converges absolutely in $\|\cdot\|_c$, so
$\|F_{\text{bad}}\|_c\le\sum_\gamma\|B_\gamma\|_c$ by the triangle inequality. By \cref{ass:cleaning} every $B_\gamma$
carries a conversion polymer of cleaned length $\ge L(d)$. Grouping terms by polymer length and using
\cref{eq:c5-absolute},
\begin{equation}
  \|F_{\text{bad}}\|_c\le\sum_{\ell\ge L(d)}\sum_{\gamma:\,|\Gamma_\gamma|=\ell}\|B_\gamma\|_c
  \le C_0C_1\,|V(G_d)|\sum_{\ell\ge L(d)}(\mu_G z)^\ell
  =\frac{C_0C_1}{1-e^{-\tau}}\,|V(G_d)|\,e^{-\tau L(d)} .
\end{equation}
Combining with \cref{eq:norm-bound} of \cref{lem:norm} gives, with $C=C_0C_1/(1-e^{-\tau})$,
\begin{equation}
  \Dstab(\Ehat)\,\pacc\le C\,|V(G_d)|\,e^{-\tau L(d)} .
\end{equation}
Since $|V(G_d)|\le C_{\mathrm{vol}}d^{\,r}$ in the bounded spacetime model, this is \cref{eq:tradeoff}. If
$L(d)=\Omega(d)$ and $\Dstab(\Ehat)\ge c>0$, then $\pacc\le (C/c)\,|V(G_d)|\,e^{-\tau L(d)}=e^{-\Omega(d)}$, the
polynomial volume prefactor being beaten by the exponential. Without the polynomial-volume hypothesis, that is for an
unbounded-depth family, the prefactor need not be subexponential and the conclusion can fail.
\end{proof}

\begin{remark}[Why the absolute form]\label{rem:absolute}
Subcriticality must be stated in the unnormalised form \cref{eq:c5-absolute}. If instead one assumed the accepted
weights obeyed $\sum_\Gamma|w_A(\Gamma)|\le C d^{\,r}e^{-\tau L}\pacc$, with a factor $\pacc$ on the right, the same
steps would give the stronger normalised bound $\Dstab(\Ehat)\le C d^{\,r}e^{-\tau L}$, with no acceptance trade-off at
all. That statement is much closer to assuming the conclusion, so it is not the right hypothesis; the acceptance
trade-off of \cref{thm:tradeoff} is a consequence of the absolute-activity form and \cref{lem:norm}, not of a
normalised input.
\end{remark}

\begin{remark}[A worked illustration]\label{rem:worked-illustration}
The assembly is transparent on a toy model that keeps only the conversion sector on a single annular ring of $n$ cells.
Let each cell carry an independent conversion amplitude of coefficient-norm at most $z$, and let a bad term $B_\Gamma$
be a product of conversions along a connected arc $\Gamma$, so $\|B_\Gamma\|_c\le z^{|\Gamma|}$. The arcs of length
$\ell$ number at most $n$ (one for each starting cell, in one orientation), and the spanning arcs are those with
$\ell\ge L$. Then
\begin{equation}
  \|F_{\text{bad}}\|_c\le\sum_{\ell\ge L}\sum_{|\Gamma|=\ell}\|B_\Gamma\|_c
  \le n\sum_{\ell\ge L}z^\ell
  =\frac{n\,z^L}{1-z},
\end{equation}
which for $z<1$ is exponentially small in $L$ with rate $\tau=-\log z$ and a prefactor linear in the system size. The
general bound replaces the arc count $n$ by the spacetime count $N_\ell(G_d)\le C_1|V(G_d)|\mu_G^{\,\ell}$ of
\cref{lem:polymer-counting} and the rate $-\log z$ by $-\log(\mu_G z)$, which is positive exactly when the gas is
subcritical. The illustration shows that the only place the argument can fail is the sign of $\mu_G z-1$, which is
\cref{cond:subcrit}, and that everything else is bookkeeping.
\end{remark}

We record the honest label. \Cref{thm:tradeoff} is a theorem under
\cref{ass:decomp,ass:cleaning}, \cref{cond:subcrit}, and \cref{lem:norm}. Of these, only \cref{lem:norm} is proved
here; \cref{ass:decomp,ass:cleaning} are structural assumptions on the accepted-effect expansion, and \cref{cond:subcrit}
is the physical assumption whose status is the subject of \cref{sec:open} and \cref{app:linetension}. The
unconditional core of the no-go is \cref{thm:iso} together with the normalisation identity \cref{lem:norm}; the
structural assumptions and \cref{cond:subcrit} enter only for the polymer extension to the coarse-grained family.

\section{The accepted conversion line-tension conjecture}\label{app:linetension}

Subcriticality (\cref{cond:subcrit}) is the one hypothesis of \cref{thm:tradeoff} we cannot remove. This appendix
states precisely what it would take to remove it, shows why the natural argument fails, and describes the construction
that comes closest to violating it. The outcome is that subcriticality is neither derivable from the other conditions
nor refutable by any construction we have found: it is a genuine open input, and the honest form of the result keeps
it as a stated conjecture.

\subsection{The conjecture}\label{sec:appC-conjecture}

Removing \cref{cond:subcrit} from \cref{thm:tradeoff} means deriving it from the operational conditions. That is the
following statement.

\begin{conjecture}[Accepted conversion line tension]\label{conj:linetension}
Every resource-free (\cref{cond:free}), $\Omega(d)$-protected (\cref{cond:protect}) accepted family with
non-negligible acceptance (\cref{cond:accept}) satisfies subcriticality (\cref{cond:subcrit}): its accepted per-cell
conversion activity $z_A$ obeys $\mu_G z_A<1$ uniformly in $d$, equivalently the accepted transfer operator has a uniform
gap between the native and the conversion sector.
\end{conjecture}

If \cref{conj:linetension} held, \cref{cond:subcrit} would be a consequence of \cref{cond:free,cond:protect,cond:accept}.
The no-go would then rest on only the two structural assumptions \cref{ass:decomp,ass:cleaning}, whose charged-polymer
and cleaning content the local-noise reduction of \cref{app:localclass} supplies for local-noise stabilizer families,
so for any family in which those assumptions hold
\cref{cond:magic,cond:free,cond:protect,cond:accept} alone would be incompatible and any resource-free protected block
that measured the magic axis sharply would be forced to exponentially small acceptance. This is the strongest form of
the resource-necessity statement we could hope for, and it is the form we do not have.

\subsection{Why distance does not give line tension}\label{sec:appC-obstruction}

The reason the conjecture is hard is a mismatch of type. Distance is a minimum-length statement, and subcriticality is
a free-energy statement. The accepted weight of conversion polymers of length $\ell$ scales as
$\sum_{|\Gamma|=\ell}|w_A(\Gamma)|\sim(\mu_G z_A)^\ell$, the number $\mu_G^{\,\ell}$ of strings times the per-string weight
$z_A^\ell$. Protection controls only the lower limit $\ell\ge L(d)$ of the sum, and says nothing about whether the
summand grows or decays, which is the sign of $\mu_G z_A-1$. Local positivity, which for a valid effect $0\le F\le I$
gives the per-cell bound $z_A\le1$, is also silent, since with $\mu_G>1$ the product $\mu_G z_A$ can still exceed one. The
three inputs we have therefore constrain different features of the same sum: distance sets where it starts, positivity
caps each term, and neither touches the ratio that decides convergence. To force convergence one needs a bound of the
form $z_A<1/\mu_G$, a constant strictly below the positivity cap, and nothing in the hypotheses supplies a constant gap
below one.

\begin{proposition}[Acceptance gives only subexponential reweighting]\label{prop:obstruction}
The acceptance handle bounds the per-length reweighting of a conversion polymer only subexponentially: for a spanning
polymer $\Gamma$ with $|\Gamma|=\ell=\Omega(d)$, non-negligible acceptance $\pacc\ge d^{-q}$ gives
\begin{equation}\label{eq:acceptance-subexp}
  \frac{\Pr(A\mid\Gamma)}{\Pr(A)}\le\frac{1}{\pacc}\le d^{\,q}=e^{\,o(\ell)} ,
\end{equation}
whereas turning the positivity worst case $z_A=1$ into a subcritical activity would require a constant per-cell penalty
$\Pr(A\mid\Gamma)\le\pacc\,e^{-a\ell}$ with $a>\log\mu_G$. Neither positivity ($z_A\le1$), protection (which fixes only
$\ell\ge L(d)$), nor acceptance supplies that penalty by this argument.
\end{proposition}

\begin{proof}
The reweighting bound \cref{eq:acceptance-subexp} is Bayes' rule: $w_A(\Gamma)\sim w_0(\Gamma)\,\Pr(A\mid\Gamma)/\Pr(A)$,
and $\Pr(A\mid\Gamma)\le1$ while $\Pr(A)=\pacc\ge d^{-q}$. Since $d^q=e^{q\log d}$ and $\ell=\Omega(d)$, the exponent
$q\log d$ is $o(\ell)$, so the reweighting is subexponential in the polymer length. Subcriticality needs the accepted
activity below $z_c=1/\mu_G<1$ (a constant strictly below one), which by
$z_A^\ell\sim z_0^\ell\,\Pr(A\mid\Gamma)/\Pr(A)$ requires the reweighting to supply a factor $e^{-a\ell}$ with
$a>\log\mu_G$ per unit length. A subexponential factor cannot supply a constant-per-length penalty, so \cref{cond:accept}
does not close the gap; neither does positivity ($z_A\le1$) nor protection (which fixes only $\ell\ge L(d)$). Two
routes that appear to help both reduce to the conjecture. A Kramers--Wannier or $\ee\!\leftrightarrow\!\mm$ duality can
identify the critical surface but not on which side the accepted coupling lies, and at the self-dual point the transfer
gap closes; asserting the high-temperature side is asserting $\mu_G z_A<1$. A Dobrushin uniqueness condition~\cite{Dobrushin1968}
$\sup_x\sum_y D_{xy}<1$ for the accepted conditional measure is, for this gas, the statement $\mu_G z_A<1$ itself. In
both cases the input equals the conclusion.
\end{proof}

Standard cluster-expansion, Pirogov--Sinai, and percolation methods prove exponential decay \emph{from} a subcritical
activity~\cite{KoteckyPreiss1986,FernandezProcacci2007,PirogovSinai1975,AizenmanBarsky1987,Grimmett1999}; none derives
the subcritical activity from distance.
A proof of \cref{conj:linetension} would need a new monotonicity or duality argument specific to accepted stabilizer
cobordisms, one that shows post-selection cannot reweight the conversion gas onto the critical side without an excluded
resource.

\subsection{Why the duality and Dobrushin routes reduce to the conjecture}\label{sec:appC-routes}

We spell out the two routes that appear to supply the missing suppression, since each is a natural first idea and each
turns out to assume its own conclusion. Write $Z_{\mathrm{native}}(L,w;z_A)$ and $Z_{\mathrm{conv}}(L,w;z_A)$ for the
accepted polymer partition functions on a length-$L$, width-$w$ annulus in the native and conversion sectors, so that
$Z_\bullet=\langle\bullet|T_\bullet(z_A)^L|\bullet\rangle$ for transfer operators $T_{\mathrm{native}}$,
$T_{\mathrm{conv}}$. The conversion line tension is the free-energy gap
\begin{equation}
  \tau(z_A)=-\lim_{L\to\infty}\frac1L\log\frac{Z_{\mathrm{conv}}}{Z_{\mathrm{native}}}
  =\log\frac{\lambda_{\mathrm{native}}(z_A)}{\lambda_{\mathrm{conv}}(z_A)},
\end{equation}
with $\lambda_\bullet$ the leading transfer eigenvalues, so \cref{cond:subcrit} is exactly
$\lambda_{\mathrm{conv}}(z_A)<\lambda_{\mathrm{native}}(z_A)$, that is $\tau(z_A)>0$.

A Kramers--Wannier or $\ee\!\leftrightarrow\!\mm$ duality identifies the dual coupling and the critical surface, but it
does not fix which side of that surface the accepted ensemble occupies. A high-temperature expansion of the dual model
gives $Z_{\mathrm{conv}}/Z_{\mathrm{native}}\le C\,|V(G_d)|\sum_{\ell\ge L}(\mu_G z_A^*)^\ell$ only under the high-temperature
input $\mu_G z_A^*<1$, and at the self-dual point the two eigenvalues meet, $\lambda_{\mathrm{conv}}=\lambda_{\mathrm{native}}$,
so $\tau=0$ and the exponential suppression is false. Asserting that duality maps the accepted conversion sector to a
native high-temperature expansion therefore already assumes the accepted coupling is on the high-temperature side,
which is \cref{cond:subcrit}, not a derivation of it.

The Dobrushin route is the same input in a different form. Assume for it that the accepted absolute conversion weights
define a positive finite-range Gibbs specification $\nu_A$, and define the influence matrix
\begin{equation}
  D_{xy}=\sup_{\eta,\eta':\,\eta|_{\Lambda\setminus\{y\}}=\eta'|_{\Lambda\setminus\{y\}}}
  \bigl\|\nu_A(\omega_x\in\cdot\mid\eta)-\nu_A(\omega_x\in\cdot\mid\eta')\bigr\|_{\mathrm{TV}} .
\end{equation}
Dobrushin uniqueness~\cite{Dobrushin1968} requires $\alpha_D=\sup_x\sum_y D_{xy}<1$. Since changing the boundary at $y$
influences $x$ only through a chain of conversion polymers from $y$ to $x$, the disagreement-path estimate gives
$D_{xy}\le C_D\sum_{\Gamma:\,x\leftrightarrow y}z_A^{|\Gamma|}$, and summing over $y$ with the self-avoiding-walk count
of \cref{lem:polymer-counting},
\begin{equation}
  \sum_y D_{xy}\le C_D C_{\mathrm{SAW}}\sum_{\ell\ge1}\ell\,(\mu_G z_A)^\ell
  =C_D C_{\mathrm{SAW}}\frac{\mu_G z_A}{(1-\mu_G z_A)^2},
\end{equation}
which is finite only for $\mu_G z_A<1$. So the Dobrushin condition first needs the same subcritical path-summability as
\cref{cond:subcrit}. Both routes prove decay from a high-temperature hypothesis, and neither derives $\mu_G z_A<1$ from
resource-freeness, protection, or acceptance.

\subsection{The closest would-be violation}\label{sec:appC-attack}

The most serious attempt to refute \cref{conj:linetension} post-selects on a spanning gas of conversion strings. Model
the annular conversion cells as a graph $\mathcal A_d$, open each cell independently with probability $q$, and let
$\Phi_d$ be the maximum number of pairwise edge-disjoint open paths spanning the annulus.

\begin{proposition}[Supercritical conversion gas as a would-be refutation]\label{prop:percolation-attack-full}
For fixed $q>p_c$ there are constants $\nu(q)>0$ and $c(q)>0$ such that, for every $0<\kappa<\nu(q)$, the event
$A_d=\{\Phi_d\ge\kappa d\}$ satisfies $\pacc=\Pr_q(A_d)\ge1-e^{-c(q)d}$.
\end{proposition}

\begin{proof}
This is the standard supercritical crossing estimate for planar percolation~\cite{Grimmett1999}. By Menger's theorem
$\Phi_d$ equals the value of the unit-capacity open-edge maximum flow between the two annular boundary components, and
the minimum open capacity of a cut separating them. For $q>p_c$ the supercritical phase has a positive flow constant
$\nu(q)>0$ across annuli of fixed aspect ratio, so $\Phi_d/d\to\nu(q)$ with an exponentially small lower tail. Choosing
$\kappa<\nu(q)$ gives $\Pr_q(\Phi_d<\kappa d)\le e^{-c(q)d}$.
\end{proof}

The selector is $\Omega(d)$-robust. If $\omega\in A_d$, choose $\kappa d$ edge-disjoint open spanning paths; changing
one open edge to closed destroys at most one of them, so removing all spanning open paths needs at least $\kappa d$
edge changes.

\begin{remark}[Selector-distance convention]\label{rem:selector-distance-convention}
The Boolean event $A_d=\{\Phi_d\ge\kappa d\}$ can be Hamming distance one from its complement at configurations with
$\Phi_d=\lceil\kappa d\rceil$. For a linear decision margin, buffer the event to $\{\Phi_d\ge2\kappa d\}$ and treat
fewer than $\kappa d$ changes as protected; if the microscopic bits are cells, replace edge-disjoint by cell-disjoint.
\end{remark}

This still does not refute the no-go, for two distinct reasons. First, if an open cell applies a local
charge-converting insertion, its microscopic instrument has a nonzero matrix element $P_\eps K_{\mathrm{cell}} P_\mm\ne0$
or $P_\mm K_{\mathrm{cell}} P_\eps\ne0$, where $P_a$ is the local charge-sector projector. This is a local
$\mm\!\leftrightarrow\!\eps$ converter, forbidden by \cref{def:resourcefree}(iii), so the family gives up
\cref{cond:free}. Second, if every cell is charge-non-converting, each local Kraus operator is block diagonal in the
local charge sectors, $P_a K_{\mathrm{cell}} P_b=0$ for $a\ne b$, and a product of such cells creates no coherent
$\mm\!\leftrightarrow\!\eps$ matrix element by itself. In the product-filter version one manufactures the alignment by
post-selecting a favourable outcome on $N_d=\Theta(d^2)$ cells. If each favourable local outcome has probability at
most $\rho<1$, independence gives $\pacc\le\rho^{N_d}=e^{-\Theta(d^2)}$, failing \cref{cond:accept}. Accepting with high
probability instead leaves the coarse-grained effect a native mixture with no protected coherent conversion of order
$\pacc$, failing \cref{cond:magic}.

Finally, even a supercritical gas proves only a large absolute conversion sum $\sum_\Gamma|w_A(\Gamma)|$, not coherent
alignment. Projected onto the logical $\Hxy$ direction the magic coefficient is $M_A=\sum_\Gamma\eta_\Gamma w_A(\Gamma)$
with $|\eta_\Gamma|=1$, and a sharp magic effect needs $|M_A|=\Omega(\pacc)$, while supercriticality bounds only the
absolute sum and random transcript phases can cancel $M_A$. The missing step is a protected, charge-non-converting
mechanism forcing the terms $\eta_\Gamma w_A(\Gamma)$ to align, and no such mechanism follows from the percolation
estimate. It is exactly what \cref{thm:iso} forbids for the fine branches the family coarse-grains.

The construction is therefore the best available \emph{attack surface}, not a refutation. Together with
\cref{prop:obstruction}, it leaves \cref{conj:linetension} open in both directions. The enumeration of \cref{prop:zc}
gives a spatial-strip proxy $z_c^{\perp}(d)=1/\mu_{\perp}(d)\approx0.40$--$0.52$ for the entropy scale. This is an
illustrative target, not the operative threshold: the object the theorem counts is the spacetime connective constant
$\mu_G$, and the two-dimensional enumeration calibrates only a lower-dimensional proxy for it.

\subsection{What a proof would have to establish}\label{sec:appC-necessary}

The two failure modes of the percolation attack say something positive about any eventual proof, so we record it. A
proof of \cref{conj:linetension} does not need to control the full accepted instrument. By \cref{lem:norm} and the
assembly of \cref{app:thm2}, it needs only a single scalar bound on the accepted conversion activity, namely a constant
$\delta>0$, independent of $d$, with
\begin{equation}\label{eq:activity-target}
  z_A\le\frac{1}{\mu_G}-\delta .
\end{equation}
Everything upstream of this is already in place: the decomposition and cleaning assumptions supply the polymer
expansion, and the counting lemma supplies $\mu_G$. So the entire open problem is to bound one effective coupling of the
accepted ensemble below the critical value.

The percolation attack shows what such a bound must survive. It must hold for an ensemble whose acceptance is chosen
adaptively on the syndrome, since that is the only freedom a decoder has, and it must remain true when the acceptance
event is a complicated decoded region and not a simple product filter. A bound that used only the marginal
per-cell conversion probability would not survive, because the attack tunes exactly that marginal toward criticality;
what is needed is a bound that uses protection or resource-freeness to constrain the \emph{correlations} of the
accepted ensemble, not just its marginals. Concretely, one wants a statement of the form
\begin{equation}
  \Pr(A\mid \Gamma)\le \Pr(A)\,e^{-a|\Gamma|},\qquad a>\log\mu_G,
\end{equation}
which \cref{prop:obstruction} shows does not follow from acceptance alone, but which a monotonicity of the accepted
measure under lengthening a conversion string would give. The natural candidate is a stochastic-domination or
disagreement-percolation argument that compares the accepted conversion measure to a genuinely subcritical reference
measure supported away from the excluded resources. We have not found such a comparison, and we suspect it needs an
input specific to stabilizer cobordisms, some form of the fact that a resource-free spacetime has no place to store the
long-range correlation a critical ensemble requires. That input, if it exists, is the missing lemma.

We close with the honest bottom line on the conjecture, stated as three facts. It is not proved: none of protection,
positivity, or acceptance forces the accepted activity below the critical value, and the two arguments that look as
though they should, duality and Dobrushin uniqueness, each assume the subcritical side they are supposed to establish.
It is not refuted: the construction that would build a supercritical accepted gas either uses the forbidden local
converter, or pays an acceptance cost exponential in the area, or produces incoherent weight that the fine-grained
theorem forbids from aligning. And it is sharp: the whole question is a single scalar inequality
$z_A<1/\mu_G$ for the accepted spacetime gas; the spatial-strip proxy of \cref{app:numerics} gives an illustrative
entropy scale that decreases with the code distance, while the exact spacetime value is left to a proof. A conditional no-go that isolates its one open assumption to this degree is, in our view, the
right form for the result to take until the statistical-mechanics input is supplied, and we have tried throughout to
keep
the boundary between what is proved and what is assumed exactly where the mathematics puts it, not where a cleaner
headline would prefer.

\section{Spatial-strip calibration of the entropy scale}\label{app:numerics}

Subcriticality (\cref{cond:subcrit}) asks whether the accepted spacetime charge-conversion gas has $\mu_G z<1$, with
$\mu_G$ the connective constant of the accepted spacetime cell graph $G_d$ of \cref{app:thm2}. Neither $\mu_G$ nor the
activity $z$ is settled analytically. This appendix does not compute $\mu_G$. It computes a lower-dimensional
\emph{proxy}: the connective constant $\mu_{\perp}(d)$ of self-avoiding walks on the two-dimensional annular strip, which
fixes an illustrative entropy scale $z_c^{\perp}(d)=1/\mu_{\perp}(d)$. The proxy is a calibration of the open problem,
not a bound on it: it gives a sense of the scale the missing lemma must reach and confirms that local positivity alone
does not reach it, but the operative threshold is $1/\mu_G$ for the full spacetime gas, which the two-dimensional
enumeration does not determine.

\paragraph{Method.} On the spatial slice, the conversion polymers restrict to connected charge-conversion strings on
the annulus that wraps the logical operator, whose entropy is the number of self-avoiding walks on the corresponding
lattice strip, with growth rate the connective constant $\mu_{\perp}$~\cite{MadrasSlade1993}. We take a width-$w(d)$
strip of the square lattice, with $w(d)$ the annulus cross-section of the distance-$d$ patch (we report $w=d$), and
count self-avoiding walks from a fixed origin by deterministic depth-first enumeration up to length $n_{\max}$,
estimating $\mu_{\perp}(d)$ from the ratio $c_{n+1}/c_n$ of walk counts. The enumeration is exact and seed-free; the
full script is released (see the code availability statement). The spacetime graph $G_d$ adds the time direction and has
a larger connective constant, so $\mu_{\perp}(d)$ is a proxy for the entropy scale and not the value of $\mu_G$.

\begin{estimate}[Numerical, spatial-strip proxy]\label{prop:zc}
On the width-$w=d$ annulus strips, the ratio estimate $\mu_{\perp}(d)$ of the spatial conversion-polymer connective
constant, taken as the mean of the last three finite-length ratios $c_{n+1}/c_n$, and the induced proxy activity
$z_c^{\perp}(d)=1/\mu_{\perp}(d)$, are
\begin{center}
\begin{tabular}{cccccc}
\toprule
$d$ & $3$ & $5$ & $7$ & $9$ & ($w\!\to\!\infty$)\\
\midrule
$n_{\max}$       & $20$ & $18$ & $17$ & $16$ & --- \\
$\mu_{\perp}(d)$         & $1.918$ & $2.232$ & $2.407$ & $2.477$ & $2.638$ \\
$z_c^{\perp}(d)=1/\mu_{\perp}(d)$ & $0.521$ & $0.448$ & $0.415$ & $0.404$ & $0.379$ \\
\bottomrule
\end{tabular}
\end{center}
The ratio estimates rise monotonically toward the square-lattice value $\mu_{\mathrm{SAW}}\approx2.638$; the residual
finite-$n_{\max}$ bias is of order the last ratio increment, $\lesssim0.05$, and is an overestimate of $z_c^{\perp}$
(the true $\mu_{\perp}$ is slightly larger), so it does not weaken the reading below. These are numerical estimates on a
finite two-dimensional lattice, not analytic bounds, and they proxy the spatial entropy scale only, not the spacetime
$\mu_G$.
\end{estimate}

\paragraph{Reading.} Three facts follow, read as statements about the spatial proxy. First, $\mu_{\perp}(d)>1$ for
every $d$, so $z_c^{\perp}<1$ strictly: on the spatial strip the entropy already puts the critical activity below
$z_c^{\perp}\approx0.40$--$0.52$, well inside the interval that local positivity leaves open. Second, local positivity
(\cref{app:thm2}) gives only $z\le1$; at the worst case $z=1$ the line tension against this proxy is
$\tau=-\log\mu_{\perp}<0$, supercritical. Positivity therefore does not deliver subcriticality even against the
lower-dimensional proxy, and a fortiori not against $\mu_G$. Third, $z_c^{\perp}(d)$ decreases with $d$, so the
constraint tightens for larger codes. Whether the accepted, post-selected activity $z_A$ of a resource-free protocol
falls below the operative threshold $1/\mu_G$ is the instrument-level input that \cref{app:linetension} shows is not
fixed by resource-freeness and protection alone.

The measurement is deliberately modest, and its role is to calibrate the scale of the open problem, not to settle it. It
computes the entropy of the two-dimensional spatial slice, a proxy that a finite-strip enumeration converges to; it does
not compute the spacetime connective constant $\mu_G$, which carries the time direction and is larger. Two things
therefore remain for a proof: the operative entropy $\mu_G$ of the accepted spacetime graph, and the energy side, the
accepted activity $z_A$ that depends on the instrument and is what the missing lemma must control. We do not claim the
connective constant is settled; the spatial proxy only says that the window local positivity leaves open is already
nonempty on the slice, and that it narrows with $d$. At $d=9$ the proxy threshold is near the square-lattice value, so
even the slice demands a constant gap under the positivity bound of one. That is the sense, a modest one, in which the
numerics sharpen the conjecture.

\paragraph{Deterministic enumeration.} For the finite-width probe, replace the annulus locally by the square-lattice
strip $S_w=\mathbb Z\times\{0,1,\ldots,w-1\}$. A self-avoiding-walk state at depth $n$ is a pair $(v_n,V_n)$, the
current endpoint $v_n\in S_w$ and the visited set $V_n$ with $|V_n|=n+1$. The walk starts at $v_0=(0,\lfloor w/2\rfloor)$,
and the allowed moves are the four nearest-neighbour steps that stay in the strip and avoid $V_n$. Depth-first search
enumerates all extensions and increments the count $c_w(n)$ at each depth. The procedure is deterministic and seed-free.
The finite-$n$ ratio estimator is $\widehat\mu_w(n)=c_w(n)/c_w(n-1)$. For fixed $w$ the strip has a transfer-matrix
description, so $c_w(n)=A_w\mu_w^n+O(\mu_{w,2}^n)$ with $\mu_{w,2}<\mu_w$, and $\widehat\mu_w(n)\to\mu_w$.

\paragraph{Transfer-matrix route.} The same connective constant follows from a transfer matrix on the width-$w$ strip,
which is a second, independent estimate of the line tension. A transfer state records the connectivity pattern of
occupied conversion-polymer edges crossing a vertical cut, together with the sector label, native for trivial annular
charge and conversion for one protected $\mm\!\leftrightarrow\!\eps$ defect threading the annulus. Adding a column gives
transfer operators $T_{\mathrm{native}}(z)$ and $T_{\mathrm{conv}}(z)$, each occupied conversion cell contributing
activity $z$, with leading eigenvalues $\lambda_{\mathrm{native}}(z)$ and $\lambda_{\mathrm{conv}}(z)$. The transfer-gap
estimator for the line tension is $\widehat\tau_w(z)=\log[\lambda_{\mathrm{native}}(z)/\lambda_{\mathrm{conv}}(z)]$, and
the critical activity is where the gap closes, $\lambda_{\mathrm{conv}}(z_c)=\lambda_{\mathrm{native}}(z_c)$. For the
pure conversion gas the generating function $C_w(z)=\sum_n c_w(n)z^n$ has radius of convergence
$z_c(w)=1/\lim_n c_w(n)^{1/n}=1/\mu_w$ by Cauchy--Hadamard, and the transfer matrix computes the same growth rate, so
the two routes agree: $\widehat\tau_w(z)=-\log(\mu_w z)$ and $z_c(w)=1/\mu_w$. In the full accepted model the native
background renormalises both eigenvalues, but the criterion is unchanged: subcriticality is the strict transfer gap
$\lambda_{\mathrm{conv}}(z_A)<\lambda_{\mathrm{native}}(z_A)$.

\paragraph{A width-two check.} The narrowest nontrivial strip gives a closed-form sanity check on the transfer-matrix
machinery. A directed non-reversing walk on the width-two strip has a transfer matrix whose growth rate is the golden
mean, since a step either advances along the strip or switches row, giving the recurrence $a_{n+1}=a_n+a_{n-1}$ and
\begin{equation}
  \mu_2^{\mathrm{dir}}=\frac{1+\sqrt5}{2}\approx1.618,\qquad 1/\mu_2^{\mathrm{dir}}=\frac{2}{1+\sqrt5}\approx0.618 .
\end{equation}
This is a distinct walk class from the self-avoiding walk of \cref{prop:zc}, and the two need not share a growth rate at
finite width, so we use it only to check that the transfer-matrix setup reproduces a hand-solvable case, not as a
validation of the self-avoiding enumeration. The width-three and wider strips have larger transfer matrices without a
closed form, and there the ratio estimator of \cref{prop:zc} is the practical route.

\paragraph{Finite-size and dimensional caveat.} The enumeration values are finite-$w$, finite-$n_{\max}$ estimates, not
analytic bounds. Finite width suppresses entropy relative to the plane, and finite depth leaves ratio-estimator
corrections. More importantly, the enumeration is two-dimensional, whereas the polymers of \cref{thm:tradeoff} live in
the $(2{+}1)$-dimensional spacetime graph $G_d$; the spacetime connective constant $\mu_G$ is larger than the spatial
$\mu_{\perp}$, so $z_c^{\perp}$ overestimates the operative threshold $1/\mu_G$. The reported $z_c^{\perp}(d)$ should
therefore be read as a spatial-slice calibration of the entropy scale, not as the universal threshold for the accepted
spacetime gas and not as a proof of \cref{cond:subcrit}. Even so, the trend is informative: the proxy falls from
$0.618$ at $w=2$ toward the square-lattice $0.379$, so the wider the code the smaller the activity a subcriticality
proof must reach, and the narrower the window that local positivity leaves open.

\section{A worked example: the distance-3 patch}\label{app:example}

We make the objects of the theorem concrete on the smallest nontrivial patch, the distance-3 rotated surface code
with nine data qubits. This is the setting a first cultivation attempt uses, and it shows both why the resource-free
check fails and what the fold supplies.

\paragraph{The patch and its logical strings.} The nine data qubits sit on a $3\times3$ grid. The stabilizer group is
generated by the weight-4 bulk plaquettes and vertices together with the weight-2 boundary generators, eight
independent generators in all, leaving one logical qubit. The boundaries are rough top and bottom and smooth left and
right, so an $\mm$-string can end on the smooth edges and an $\ee$-string on the rough edges. The logical $\Xbar$ is a
horizontal weight-3 $X$-string between the smooth boundaries. As an anyon string it carries an $\mm$ charge across the
patch, so $d_b(\Xbar)=\mm$. The logical $\Zbar$ is a vertical weight-3 $Z$-string between the rough boundaries, with
$d_b(\Zbar)=\ee$. Their product $\Ybar=i\Xbar\Zbar$ is the diagonal $\eps$-string, $d_b(\Ybar)=\eps$. The three strings
realise the braiding $\Braid(\mm,\ee)=\Braid(\mm,\eps)=-1$ through the single anticommutation $\Xbar\Zbar=-\Zbar\Xbar$.
Cleaning to a canonical representative, the analysis of \cref{app:thm1}, means pushing an $X$-string off the code
stabilizers to the weight-3 horizontal line and reading its homology class across the annulus, which on this patch is
simply whether it connects the two smooth boundaries.

\paragraph{Why the resource-free check fails.} Suppose we try to check $\Hxy$ without any added structure. A single
transcript of resource-free stabilizer measurements can measure $\Xbar$, or measure $\Ybar$, or interleave the two,
but by \cref{lem:collapse} its accepted effect is a stabilizer projector $\Pi_S$ with $S$ abelian. On the $d=3$ patch
one checks directly that $\Xbar$ and $\Ybar$ cannot both lie in a commuting group with a nonzero joint eigenspace:
they are conjugate logical operators, $\Xbar\Ybar=-\Ybar\Xbar$, so any $\Pi_S$ containing cleaned representatives of
both charges $\mm$ and $\eps$ annihilates the code space, exactly the $P_{\text{code}}\,p\,P_{\text{code}}=0$
mechanism of \cref{prop:isotropy}. One can see the collapse in a single line. If the transcript first measures $\Xbar$ with outcome $+$ and then $\Ybar$,
the accepted Kraus operator is $K=P_{\Ybar,+}P_{\Xbar,+}$, and its effect is
\begin{equation}
  K^\dagger K=P_{\Xbar,+}P_{\Ybar,+}P_{\Xbar,+}
  =P_{\Xbar,+}\,\frac{I+\Ybar}{2}\,P_{\Xbar,+}
  =\tfrac12\,P_{\Xbar,+},
\end{equation}
using $P_{\Xbar,+}\Ybar P_{\Xbar,+}=0$ since $\Ybar$ anticommutes with $\Xbar$. The second measurement adds nothing to
the charge content; the effect is proportional to the single stabilizer projector $P_{\Xbar,+}$, whose charge is the
isotropic $\{1,\mm\}$. The best a resource-free family can do is coarse-grain, running the $\Xbar$ measurement with
probability $\tfrac12$ and the $\Ybar$ measurement with probability $\tfrac12$ and accepting the $+$ outcome. The
accepted effect is
\begin{equation}
  \Eacc=\tfrac12 P_{\Xbar,+}+\tfrac12 P_{\Ybar,+}=\tfrac12 I+\tfrac14(\Xbar+\Ybar),
\end{equation}
with normalised Bloch vector $\bm v=(\tfrac12,\tfrac12,0)$, so
\begin{equation}
  \|\bm v\|_1=\tfrac12+\tfrac12=1,\qquad \Dstab(\Ehat)=\max(0,\|\bm v\|_1-1)=0 .
\end{equation}
The effect sits on the octahedron boundary. It is not the coherent check but a coin flip between two Pauli
measurements, and it retains which one was performed. This is the classical-mixture loophole closed by the magic
witness, made explicit on the patch. The sharp check would instead need $\bm v=(1/\sqrt2,1/\sqrt2,0)$ with
$\|\bm v\|_1=\sqrt2$, a point the coin flip can never reach because classical mixing only shrinks the $\ell_1$ norm.

\paragraph{What the fold supplies.} A fold-transversal Hadamard folds the patch along its diagonal and applies a
transversal Hadamard, implementing the boundary automorphism $\ee\leftrightarrow\mm$~\cite{Sahay2025,KobayashiZhu2023}.
On the folded lattice the two boundary types are identified, so the code becomes self-dual, and the transversal
Hadamard acts on the logical qubit as the logical Hadamard $\overline H$, exchanging $\Xbar$ and $\Zbar$. At the anyon
level this is the automorphism of $D(\Zt)$ that swaps $\ee$ and $\mm$ and fixes $\eps$.

We are careful about what this measures, because the point is easy to misstate. A single-qubit Clifford conjugation
sends a Pauli axis to a Pauli axis, and $\Hxy=(\Xbar+\Ybar)/\sqrt2$ is not a Pauli operator, so no Clifford brings
$\Hxy$ to a Pauli axis and no ordinary boundary readout of the folded code measures it. What the self-dual structure
supplies is different: it makes a \emph{non-Pauli Hermitian logical Clifford observable} measurable as a fold-supported
logical observable. The natural one is the logical Hadamard axis $H_{XZ}=(\Xbar+\Zbar)/\sqrt2$, whose two Pauli labels
carry the charges $\mm$ and $\ee$, so measuring it coherently needs the $\ee\!\leftrightarrow\!\mm$ relation that the
self-dual wall supplies and that \cref{fig:fold} draws. A single-qubit logical Clifford then relabels the Pauli pair
$(\Xbar,\Zbar)$ to $(\Xbar,\Ybar)$, carrying $H_{XZ}$ to $\Hxy$. We are precise about the status of this last step: it
is a logical Clifford-frame relabeling of the Pauli labels of the observable, sending the $\Zbar\sim\ee$ label to
$\Ybar\sim\eps$, and it is not a braided autoequivalence of the bare $D(\Zt)$ phase, which cannot exchange the boson
$\ee$ and the fermion $\eps$. The topological resource is the self-dual $\ee\!\leftrightarrow\!\mm$ wall; the subsequent
$\Zbar\mapsto\Ybar$ step is a logical Clifford-frame conversion used only to express the cultivated axis as $\Hxy$. So
the resource is the self-dual identification that permits measuring the logical Hadamard-type observable itself, not a
Pauli readout in a rotated frame, and it is what no resource-free cell can supply. Cultivation protocols realise the check at constant acceptance this
way~\cite{GidneyShuttyJones2024,Claes2025,Sahay2025,Vaknin2025}.

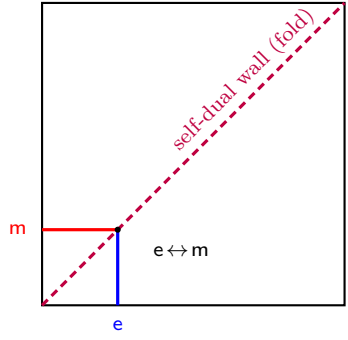
\begin{figure}[t]
\centering
\begin{tikzpicture}[scale=1.0,>=Latex]
  \draw[thick] (0,0) rectangle (4,4);
  \draw[very thick,dash pattern=on 3pt off 2pt,purple] (0,0)--(4,4);
  \node[purple,rotate=45,anchor=south] at (2.85,2.8) {\footnotesize self-dual wall (fold)};
  \draw[red,very thick] (0,1)--(1,1);
  \node[red,anchor=east] at (-0.08,1) {\footnotesize $\mm$};
  \draw[blue,very thick] (1,1)--(1,0);
  \node[blue,anchor=north] at (1,-0.08) {\footnotesize $\ee$};
  \fill (1,1) circle (0.04);
  \node[anchor=west] at (1.35,0.72) {\footnotesize $\ee\!\leftrightarrow\!\mm$};
\end{tikzpicture}
\caption{The self-dual wall of the fold. It exchanges the two boundary types, so an $\mm$-string crossing it continues
as an $\ee$-string, the automorphism $\ee\leftrightarrow\mm$. This is the coherence that a fold-supported measurement of
the non-Pauli logical Hadamard observable $H_{XZ}=(\Xbar+\Zbar)/\sqrt2$ needs; a single-qubit logical Clifford-frame
relabeling of the Pauli labels ($\Zbar\mapsto\Ybar$) then expresses the cultivated axis as $\Hxy$. A resource-free patch
has no such wall.}
\label{fig:fold}
\end{figure}

The fold is the anyon-permuting structure that \cref{def:resourcefree}(i) excludes, drawn in \cref{fig:fold}. On the
$d=3$ patch it is a concrete, weight-$O(d)$ operation, not a local gate, which is why it is a resource and not a free
cell. Its cost also
matches the theorem: the conversion string it supplies has weight of order $d$, exactly the spanning length that
\cref{app:thm2} shows any accepted conversion must have. So the anyon-permuting, self-dual resource that the fold
supplies is necessary under the assumptions of \cref{thm:main}, and the fold construction exhibits one sufficient way to
supply it; this closes the tightness picture of \cref{sec:results-roles}. The necessity is of the resource type, not of
the fold specifically: a cross-cap, a code switch, or another anyon-permuting defect supplies the same resource, and the
theorem says this necessity is not an accident of the $d=3$ construction. Within the resource-free class, and under the
assumptions of \cref{thm:main}, some such resource is needed at every distance.

\paragraph{The general lesson from the small case.} The distance-3 patch is small enough to hold in one's head, and it
already shows the whole argument. The magic axis needs two charges that anticommute as logical operators, a single
stabilizer transcript can carry only one of them because a commuting group cannot contain an anticommuting pair, and a
classical mixture of the two collapses onto the stabilizer boundary because mixing shrinks the $\ell_1$ norm. Every
resource-free move on the patch is one of these three, and none reaches the sharp check. The only way out on the small
patch is the fold, which physically identifies the two boundary types and so lets a single string carry both charges
in turn. Nothing in this picture used $d=3$ except the convenience of drawing it. At larger distance the strings are
longer, the stabilizer group is bigger, and the fold is a wider operation, but the algebra is the same, and the general
theorem is what promotes the small-case observation into a statement at every distance, with the one caveat that the
coarse-grained decoded family needs the assumptions of \cref{app:thm2} that a single transcript does not.

\section{The local-noise reduction and its residual hypothesis}\label{app:localclass}

It is natural to ask whether the two structural assumptions (\cref{ass:decomp,ass:cleaning}) and subcriticality
(\cref{cond:subcrit}) can be \emph{derived} instead of assumed, once the noise is local and the operations are
stabilizer. This appendix carries that reduction out honestly. Local stochastic noise does supply the polymer structure
that \cref{ass:decomp,ass:cleaning} describe. It does not supply subcriticality: under post-selection the reduction
terminates at a single explicit hypothesis on the decoder, and that hypothesis is exactly \cref{cond:subcrit} in local
form. We state what is gained and where it stops, because naming the residual precisely is more useful than absorbing it
into a class definition.

\subsection{The local-noise stabilizer setting}\label{app:localclass-setting}

Fix constants $R,\Delta,M,\lambda<\infty$, independent of $d$. Consider a protocol whose spacetime cell graph
$\mathcal G_d$ has degree at most $\Delta$ and range at most $R$, with fixed rough and smooth boundaries, built from
product stabilizer ancillas, charge-preserving local Clifford and Pauli operations, native Pauli--Wilson measurements,
and classical feed-forward, with no twist, fold, cross-cap, code switch, self-dual patch, or local coherent
$\mm\!\leftrightarrow\!\eps$ converter. This is the local, resource-free (\cref{def:resourcefree}) skeleton of a
check-and-grow protocol. Noise is local stochastic Pauli noise at rate $p$: after Pauli-frame propagation through the
ideal stabilizer circuit, the set $C(\omega)\subset V(\mathcal G_d)$ of cells carrying a conversion label obeys
$\Pr_p[\,S\subset C(\omega)\,]\le(\lambda Mp)^{|S|}$ for every finite $S$, with $M$ bounding the conversion labels per
cell. Finite-range correlated Pauli noise with the same domination is allowed.

\subsection{What local noise supplies}\label{app:localclass-gains}

For a fixed noise configuration and measurement transcript the circuit is a stabilizer circuit with Pauli faults, so
its code-space effect is a stabilizer effect and the local Pauli-noise expansion is charge-resolved cell by cell,
\[
   \mathcal E_x=\mathcal E_x^{\mathrm{nat}}+\sum_{\alpha\in\mathrm{conv}(x)}\lambda_{x,\alpha}(p)\,\mathcal E_{x,\alpha},
   \qquad
   \sum_{\alpha\in\mathrm{conv}(x)}|\lambda_{x,\alpha}(p)|\le c_{\mathrm{op}}\lambda Mp .
\]
Terms whose cleaned charge content is isotropic are positive combinations of one-qubit stabilizer effects; grouping
them gives a candidate $F_{\mathrm{native}}$ and a remainder organised over connected conversion polymers, the
candidate $F_{\mathrm{bad}}=\sum_\gamma B_\gamma$. This is the structural content of \cref{ass:decomp}: local noise
produces a genuine charged-polymer organisation with per-cell activity $O(p)$. The minimum-length content of
\cref{ass:cleaning} follows in the same setting from the cleaning lemma of \cref{app:thm1}: because the spacetime is a
product $D(\Zt)$ cobordism with fixed boundaries and no wall, a nontrivial annular $\mm\!\leftrightarrow\!\eps$ label
must be carried by a connected chain of conversion cells, and a chain of cleaned length below $L(d)$ lies in a
correctable region, so its code action is native and cannot change the annular charge. A genuine conversion term
therefore has cleaned length at least $L(d)$.

Two caveats keep this honest. The split $F_{\mathrm{native}}\in\mathcal N$ is a positivity statement, and a signed
cluster expansion can spoil the positivity of the native remainder unless the accepted weights are controlled; and the
polymer organisation is canonical only before post-selection reweights the transcripts. Both caveats point to the same
place: the accepted weight of a polymer, not its bare noise weight, is what \cref{cond:subcrit} must bound.

\subsection{Where the reduction stops}\label{app:localclass-residual}

The activity that \cref{cond:subcrit} counts is the \emph{accepted} activity. Writing $a_d(\omega)\in[0,1]$ for the
acceptance weight, subcriticality asks that
\begin{equation}\label{eq:cbl}
   \Pr_p[\,S\subset C(\omega)\mid A\,]
   =\frac{1}{\pacc}\sum_{\omega:\,S\subset C(\omega)}\Pr_p(\omega)\,a_d(\omega)
   \;\le\;(\Lambda\lambda Mp)^{|S|}
\end{equation}
for a constant $\Lambda$ independent of $d$: conditioning on acceptance may reweight a fixed conversion set by at most a
constant per cell. This \emph{conditional bounded-likelihood} property is what makes the accepted activity $z_A=O(p)$
and hence $\mu_G z_A<1$ at small $p$, closing the reduction. It is also exactly the point that local noise does not give
for free. The bare bound $\Pr_p[S\subset C(\omega)]\le(\lambda Mp)^{|S|}$ is unconditional and holds automatically, but
it does not control the conditional ensemble \cref{eq:cbl}: since $a_d\le 1$,
\[
   \sum_{\omega:\,S\subset C}\Pr_p(\omega)\,a_d(\omega)\le\Pr_p[S\subset C],
\]
which is only an unconditional numerator bound; after division by $\pacc$ the conditional probability can be larger by
any factor up to $1/\pacc$. A
decoder that accepts almost only when a conversion polymer occurs realises this: with one conversion cell of
probability $p$ and acceptance equal to its indicator, $\Pr_p[C\mid A]=1$, an enhancement of order $1/p$ over the bare
$p$. Such a decoder drives $z_A$ to order one and defeats subcriticality while keeping the noise strictly local. It is
the same escape route the general obstruction of \cref{sec:open} describes, now visible at the level of a single cell.

So the local-noise reduction is genuine but partial. It derives the charged-polymer structure of
\cref{ass:decomp,ass:cleaning} from locality and cleaning, and it reduces the entire remaining question to the single
condition \cref{eq:cbl} on the decoder: that post-selection does not reweight the conversion sector by more than a
constant per unit length. Condition \cref{eq:cbl} is \cref{cond:subcrit} for local noise, stated as a checkable property
of the acceptance rule. A complementary-gap decoder whose accept region is conversion-blind, in the precise sense that
its indicator has bounded likelihood ratio against the conversion labels, satisfies \cref{eq:cbl}; a decoder tuned to
accept on conversion does not. Which of these describes the decoder of a real cultivation protocol is open, and it is
the concrete form the line-tension conjecture (\cref{conj:linetension}) takes in the local-noise class. We record it as
a residual hypothesis, isolated and physically interpretable, not buried as a hidden clause of a class definition. The
one part of the obstruction that needs none of this is the fine-grained isotropy theorem (\cref{thm:iso}), which forbids
the coherent alignment on every single stabilizer transcript unconditionally.

\end{document}